\documentclass[a4paper,10pt]{article}
\usepackage[utf8]{inputenc}
\usepackage{amsmath,amssymb,amsfonts,amsthm}
\usepackage{bbm}
\usepackage{graphicx}
\usepackage{subfig}
\usepackage{multirow}
\usepackage[pdftex,bookmarks=false,colorlinks=true,linkcolor=black,
citecolor=blue,filecolor=blue,urlcolor=blue]{hyperref}

\numberwithin{equation}{section}

\providecommand{\ud}{\,\mathrm{d}}
\providecommand{\abs}[1]{\left\lvert#1\right\rvert}

\providecommand{\ket}[1]{\left\vert#1\right\rangle}
\providecommand{\hprod}[2]{\left\langle#1\,\vert\,#2\right\rangle}
\providecommand{\bprod}[2]{\left(#1\,\vert\,#2\right)}

\providecommand{\vect}[1]{{\boldsymbol{#1}}}
\providecommand{\norm}[1]{\left\lVert#1\right\rVert}

\providecommand{\conj}[1]{\overline{#1}}

\providecommand{\R}{{\mathbb R}}
\providecommand{\C}{{\mathbb C}}
\providecommand{\Z}{{\mathbb Z}}
\providecommand{\N}{{\mathbb N}}
\providecommand{\ii}{\mathbbm{i}}

\newtheorem{definition}{Definition}

\newtheorem{proposition}[definition]{Proposition}

\begin{document}

\title{Efficient Algorithm for Two-Center\\
Coulomb and Exchange Integrals of\\
Electronic Prolate Spheroidal Orbitals}

\author{Christian B. Mendl\footnote{Mathematics Department, Technische Universit\"at M\"unchen; mendl@ma.tum.de}}

\date{\today}

\maketitle

\begin{abstract}
\small
\noindent We present a fast algorithm to calculate Coulomb/exchange integrals of prolate spheroidal electronic orbitals, which are the exact solutions of the single-electron, two-center Schr\"odinger equation for diatomic molecules. Our approach employs Neumann's expansion of the Coulomb repulsion $1/\abs{\vect{x}-\vect{y}}$, solves the resulting integrals symbolically in closed form and subsequently performs a numeric Taylor expansion for efficiency. Thanks to the general form of the integrals, the obtained coefficients are independent of the particular wavefunctions and can thus be reused later.

Key features of our algorithm include complete avoidance of numeric integration, drafting of the individual steps as fast matrix operations and high accuracy due to the exponential convergence of the expansions.

Application to the diatomic molecules $\mathrm{O}_2$ and $\mathrm{CO}$ exemplifies the developed methods, which can be relevant for a quantitative understanding of chemical bonds in general.
\end{abstract}

\section{Introduction}
\label{sec:Introduction}

The two-center electronic Schr\"odinger equation is a natural starting point to study diatomic molecules or chemical bonds. It is well known that it separates in prolate spheroidal coordinates. Thus, the corresponding single-electron orbitals can be calculated efficiently. For several electrons, however, the tedious inter-electron Coulomb repulsion integrals have impeded a widespread use of these orbitals so far. To alleviate these difficulties, we present an efficient algorithmic framework in this paper.

In the computational chemistry literature, LCAO (linear combination of atomic orbitals) is the most common approach to construct electronic wavefunctions for molecules. It dates back to the early days of quantum mechanics~\cite{LennardJones1929}. In the seminal paper~\cite{Boys1950}, Boys proposed Gaussian-type atomic orbitals since the necessary integrals can be explicitly evaluated. Hence they are widely used in modern computational chemistry software packages. Nevertheless, only the exact single-electron spheroidal orbitals are -- by definition -- precise for any distance of the atomic nuclei. This fact is an important advantage for studying diatomic molecules and chemical bonds.

An interesting alternative approach to diatomic molecules is the Holstein-Herring method~\cite{Holstein1952,Herring1962,Tang1991,AsymptoticsExchangeIntegralScott2004} for calculating exchange energies of $\mathrm{H}_2^+$-like molecular ions. This method has recently been extended to two-active-electron systems~\cite{ExchangeTwoElectronDiatomic2004}. However, it is not suitable for an arbitrary number of valence electrons.

Another common approach first proposed by Hylleraas~\cite{Hylleraas1929} for the helium atom includes the inter-electron distance $r_{ij}$ as independent variable into the electronic wavefunction. Thus, the pairwise electronic Coulomb cusp is handled explicitly, which reduces the number of required wavefunctions. James and Coolidge~\cite{JamesCoolidge1933} have applied this method to the $\mathrm{H}_2$ molecule using spheroidal coordinates, which still serves as starting point for modern benchmark calculations. Ref.~\cite{DiatomicHylleraas1977} contains an extension to the $\mathrm{He}_2^+$ and $\mathrm{He}_2$ molecule, and a modern review can be found in~\cite{R12MethodsReview2006}.

Ref.~\cite{GaussianMolecularX2001} is part of a series which provides an extensive discussion of Gaussian basis sets for molecular calculations, and specifically computes the total energy and dissociation energy of $\mathrm{O}_2$.

Ref.~\cite{KohnShamDiatomic2009} employs Kohn-Sham density functional theory for diatomic molecules in (discretized) spheroidal coordinates. In particular, the authors apply their method to calculate the ground state energy of carbon monoxide $\mathrm{CO}$.


The basic setup of prolate spheroidal orbitals employed in the current paper has been developed in Ref.~\cite{AubertBessis1974,AubertBessis2nd1974} and applied to molecules with up 4 electrons. Our contribution is a reformulation into an efficient computational framework\footnote{The complete source code of our implementation is available online at~\cite{FermiFabSoftware} (in the \texttt{mathematica/diatomic} subfolder)}, which allows for an extension to many more electrons. For example, the oxygen dimer $\mathrm{O}_2$ contains 16 electrons.

\paragraph{Outline}
Section~\ref{sec:Framework} provides the details of the single-electron Schr\"odinger equation for atomic dimers in prolate spheroidal coordinates. Our presentation is based on the series~\cite{AubertBessis1974,AubertBessis2nd1974}, and additionally includes a ``best match'' mapping to the common LCAO molecular orbitals. Section~\ref{sec:LaguerreExpansions} contains the main abstract mathematical contribution of this paper: we prove a recurrence relation to efficiently multiply function expansions in terms of associated Laguerre polynomials, and solve several integrals symbolically in closed form. These results (combined with Neumann's expansion of $1/\abs{\vect{x}-\vect{y}}$ into Legendre polynomials) are the basis of our algorithm. It is described in detail in section~\ref{sec:Coulomb}, including cost analysis and error estimation. Section~\ref{sec:Results} contains the application of the algorithm to the $\mathrm{O}_2$ molecule, which is particularly interesting among atomic dimers due to its paramagnetism.

\section{Single-Electron Schr\"odinger Equation for\\Atomic Dimers}
\label{sec:Framework}

This section introduces the single-electron quantum mechanical framework, which serves as starting point for the many-electron calculations in section~\ref{sec:Coulomb}. We basically follow the discussion in Ref.~\cite{AubertBessis1974,AubertBessis2nd1974}.

\paragraph{Separation in prolate spheroidal coordinates}

The single-electron, two-center Schr\"odinger equation for a $\mathrm{H}_2^+$-like molecular ion in atomic units (Born-Oppenheimer approximation) reads
\begin{equation}
\label{eq:SchroedingerSingle}
\left(-\frac{1}{2}\Delta - \frac{Z_a}{r_a} - \frac{Z_b}{r_b}\right)\psi = E \psi.
\end{equation}
Here, $r_a$ and $r_b$ denote the distances to the fixed nuclei at $(0,0,\mp R/2)$, respectively, and $Z_a, Z_b \in \N_{>0}$ the nuclear charges (see figure~\ref{fig:DimerProlateSpheroidal}). The distance $R$ between the nuclei is also called \emph{bond length} in the chemistry literature. We have omitted the repulsive interaction of the nuclei ($\tfrac{Z_a Z_b}{R}$) for now to focus on the electronic energy, but will include it into the total energy later. In what follows, we set $Z := \tfrac{Z_a + Z_b}{2}$ and $\Delta q := (Z_a - Z_b)\,R$ (w.l.o.g.\ $\Delta q \ge 0$). The homonuclear case corresponds precisely to $Z_a = Z_b \equiv Z$ and $\Delta q = 0$.
\begin{figure}[!ht]
\centering
\includegraphics[width=0.8\textwidth]{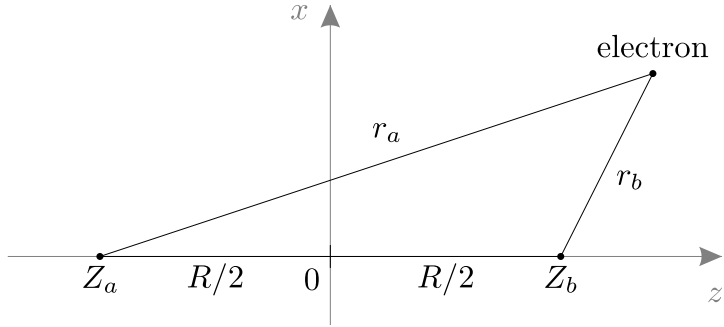}
\caption{Spatial arrangement of a single electron bound to two atomic nuclei}
\label{fig:DimerProlateSpheroidal}
\end{figure}

It is well known that equation~\eqref{eq:SchroedingerSingle} is separable in prolate spheroidal coordinates $(\xi,\eta,\varphi)$ defined by
\begin{align*}
\xi &:= \left(r_a + r_b\right)/R, \quad \xi \ge 1\\
\eta &:= \left(r_a - r_b\right)/R, \quad \eta \in [-1,1]
\end{align*}
and the Ansatz
\begin{equation}
\label{eq:Ansatz}
\psi(\xi,\eta,\varphi) = \Lambda(\xi) S(\eta) \frac{\mathrm{e}^{\ii m \varphi}}{\sqrt{2\pi}}.
\end{equation}
$m \in \Z$ is the eigenvalue of the angular momentum operator $L_z = -\ii \partial_\varphi$, which commutes with the Hamiltonian on the left hand side of~\eqref{eq:SchroedingerSingle} due to the azimuthal symmetry about the internuclear axis. In the following, we set $\mu := \abs{m}$ to shorten notation.

Plugging~\eqref{eq:Ansatz} into~\eqref{eq:SchroedingerSingle} leads to coupled ODEs for the radial part $\Lambda(\xi)$ and angular part $S(\eta)$. The latter reads
\begin{equation}
\label{eq:AngularSpheroidalODE}
\left[\frac{\partial}{\partial\eta}\left(\left(1-\eta^2\right)\frac{\partial}{\partial\eta}\right) + \underbrace{(p^2 - A)}_{\lambda^\mu_\ell(\ii p,\Delta q)} - \Delta q\,\eta + (\ii p)^2\left(1-\eta^2\right) - \frac{\mu^2}{1-\eta^2}\right] S(\eta) = 0,
\end{equation}
where the ``energy parameter'' $p \in \R_{>0}$ is defined via the energy $E$,
\begin{equation}
\label{eq:EnergySingle}
E =: -2\,(p/R)^2,
\end{equation}
and $A$ is an eigenvalue of the operator $\mathcal{G}$ (defined in~\cite{AubertBessis1974}). For the purpose of this paper, we simply regard $A$ as separation constant. In the homonuclear case $\Delta q = 0$, equation~\eqref{eq:AngularSpheroidalODE} is well know as the angular spheroidal wave equation~\cite{Stratton1941,Meixner1954,Flammer1957,Falloon2003} when we identify $(p^2 - A)$ as \emph{spheroidal eigenvalue} $\lambda^\mu_\ell(\ii p)$. Successive $\ell = \mu,\mu+1,\dots$ label the discrete set of eigenvalues for which~\eqref{eq:AngularSpheroidalODE} has a normalizable solution.

Since $E$ is finite, $p \to 0$ in the united atom limit $R \to 0$, and~\eqref{eq:AngularSpheroidalODE} reduces to Legendre's differential equation. Then $\lim_{R\to0}\mathcal{G} = -\vect{L}^2$ (angular momentum operator) with eigenvalue $A = -\lambda^\mu_\ell(0) = -\ell(\ell+1)$. However, except for this special case, $\ell$ is no valid quantum number since $\vect{L}^2$ does not commute with the Hamiltonian in general.

The homonuclear solution $S(\eta) \equiv S^\mu_\ell(\ii p,\eta)$ is already built into Mathematica\footnote{Specifically, the implementation~\cite{Falloon2003} has been integrated into Mathematica as \texttt{SpheroidalPS[n,m,$\gamma$,z]} and \texttt{SpheroidalQS[n,m,$\gamma$,z]} for the angular spheroidal function of the first and second kind, respectively.} and could thus be plugged into (numeric) integrals. Nevertheless, in order to use some properties of Legendre polynomials later and cover the heteronuclear case also, we employ the series expansion
\begin{equation}
\label{eq:SpheroidalAngularExpansion}
S^\mu_\ell(\ii p,\Delta q,\eta) = \sum_{k=\mu}^\infty c^\mu_{\ell,k}(p,\Delta q) \left(\frac{2k+1}{2} \frac{(k-\mu)!}{(k+\mu)!}\right)^{1/2} P^\mu_k(\eta).
\end{equation}
Plugged into~\eqref{eq:AngularSpheroidalODE} results in a three-term recurrence relation for the coefficients $c^\mu_{\ell,k}(p) \equiv c^\mu_{\ell,k}(p,0)$ (homonuclear) and a five-term recurrence relation for $c^\mu_{\ell,k}(p,\Delta q)$ (heteronuclear)~\cite{Meixner1954,AubertBessis1974}. Namely, in the homonuclear case, only integers $k$ with the same parity as $\ell$ contribute to the sum due to symmetry. After truncating this expansion (which is justified due to the exponential decay of the coefficients), it may be rewritten as eigenvalue equation (see also~\cite{Hodge1970})
\begin{equation}
\label{eq:AngularCoeffEquation}
F^\mu(p,\Delta q)\,\vect{c} \stackrel{!}{=} \lambda\,\vect{c},\quad \vect{c} \equiv \left(c^\mu_{\ell,k}(p,\Delta q)\right)_k, \quad \lambda \equiv \lambda^\mu_\ell(\ii p,\Delta q)
\end{equation}
with a symmetric matrix $F^\mu(p,\Delta q)$. This matrix is tridiagonal in the homonuclear case (after proper relabeling) and pentadiagonal in the heteronuclear case. Note that fast eigenvalue solvers exist particularly for tridiagonal matrices. We adopt the normalization scheme used by~\cite{Meixner1954} and Mathematica, namely
\begin{equation}
\label{eq:AngularCoeffNormalization}
\int_{-1}^1 S^\mu_\ell(\ii p,\Delta q,\eta)^2 \ud\eta = \sum_{k=\mu}^\infty \abs{c^\mu_{\ell,k}(p,\Delta q)}^2 \stackrel{!}{=} \frac{2}{2\ell+1} \frac{(\ell+\mu)!}{(\ell-\mu)!}.
\end{equation}

The energy parameter $p$ couples~\eqref{eq:AngularSpheroidalODE} to the \emph{radial} equation
\begin{equation}
\label{eq:DimerODERadialStdform}
\left[\frac{\partial}{\partial\xi}\left(\left(\xi^2-1\right)\frac{\partial}{\partial\xi}\right) - \underbrace{(p^2 - A)}_{\lambda^\mu_\ell(\ii p,\Delta q)} + 2\,Z R\,\xi + (\ii p)^2\left(\xi^2-1\right) - \frac{\mu^2}{\xi^2-1}\right] \Lambda(\xi) = 0.
\end{equation}
This is the radial spheroidal differential equation except for the $2\,Z R\,\xi$ term, and formally resembles~\eqref{eq:AngularSpheroidalODE} apart from $\xi \ge 1$ versus $\abs{\eta} \le 1$. We determine $p$ numerically as follows.

First, define \emph{Hylleraas functions} via associated Laguerre polynomials as
\begin{equation*}
H^\mu_k(x) := x^{\mu/2} \mathrm{e}^{-x/2} \sqrt{k!/(k+\mu)!}\ L^\mu_k(x), \quad k, \mu \in \N_0.
\end{equation*}
This choice precisely incorporates the orthogonality relation for Laguerre polynomials, such that
\begin{equation}
\label{eq:LaguerreOrthogonality}
\int_0^\infty H^\mu_{k'}(x) H^\mu_k(x) \ud x = \delta_{kk'}.
\end{equation}
Given a sequence $\vect{d} \equiv (d_k)_{k \ge 0}$, we set
\begin{equation}
\label{eq:HylleraasExpansion}
H^\mu_\vect{d}(x) := \sum_{k=0}^\infty d_k \,H^\mu_k(x).
\end{equation}
(Note that $k$ starts at $0$ instead of $\mu$ as in~\eqref{eq:SpheroidalAngularExpansion}.) Employing such an expansion for the radial wavefunction,
\begin{equation}
\label{eq:HylleraasWavefunction}
\Lambda(\xi) = H^\mu_{\vect{d}}(2 p\,(\xi-1))
\end{equation}
results in a three-term recursion formula~\cite{AubertBessis1974} for the to-be determined coefficients $d_k$. They will turn out to decay exponentially, as illustrated in figure~\ref{fig:DkExpansion}.
\begin{figure}[!ht]
\centering
\includegraphics[width=0.8\textwidth]{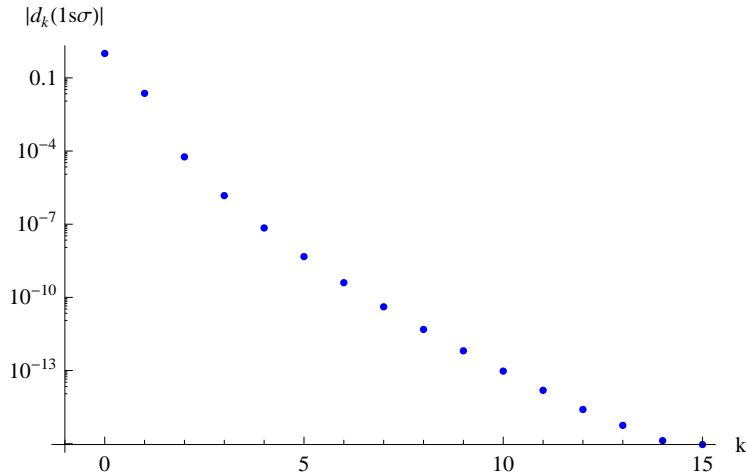}
\caption{Laguerre expansion coefficients of the $(\ell,m) = (0,0)$ groundstate radial wavefunction (see equations~\eqref{eq:HylleraasExpansion} and~\eqref{eq:HylleraasWavefunction}). The exponential decay of the coefficients justifies the truncation of the expansion.}
\label{fig:DkExpansion}
\end{figure}
Hence we can truncate the expansion and rewrite the recurrence relation as matrix equation~\cite{AubertBessis1974}
\begin{equation}
\label{eq:RadialCoeffEquation}
\left(B^\mu(p) R^\mu(p,\lambda) + p\,\mu^2 I\right)\vect{d} \stackrel{!}{=} 0, \quad \lambda \equiv \lambda^\mu_\ell(\ii p,\Delta q).
\end{equation}
Both $R^\mu$ and $B^\mu$ are symmetric tridiagonal matrices, and $I$ denotes the identity matrix. The left hand side is singular for a discrete set of values $p$ only. This condition finally determines $p$ and the energy $E$. Ref.~\cite{AubertBessis1974}~employs a Newton iteration to obtain both $p$ and $A$ simultaneously, such that the matrices in~\eqref{eq:RadialCoeffEquation} and $F^\mu - \lambda\,I$ in~\eqref{eq:AngularCoeffEquation} have zero determinants. An improved version uses the so-called Killingbeck method~\cite{Killingbeck1989,ScottAubertGrotendorst2006}. In our case, we apply a numerical root search algorithm over $p$ such that an eigenvalue of the matrix in~\eqref{eq:RadialCoeffEquation} becomes zero.

Considering the starting point of the numerical iteration, \cite{AubertBessis1974}~uses $p_0 = Z R/n$, which becomes exact in the unified atom limit $R \to 0$ and is thus valid for small $Z R$. Here, $n$ labels successive eigenvalues as in the unified atom limit. Alternatively, we have identified $p_0 = Z R/(2n)$ as reliable candidate for large values of $Z R$, which stems from the dissociation limit $R \to \infty$ (hydrogen-like atom plus isolated nucleus).

\begin{figure}[!ht]
\centering
\subfloat{
\includegraphics[width=0.9\textwidth]{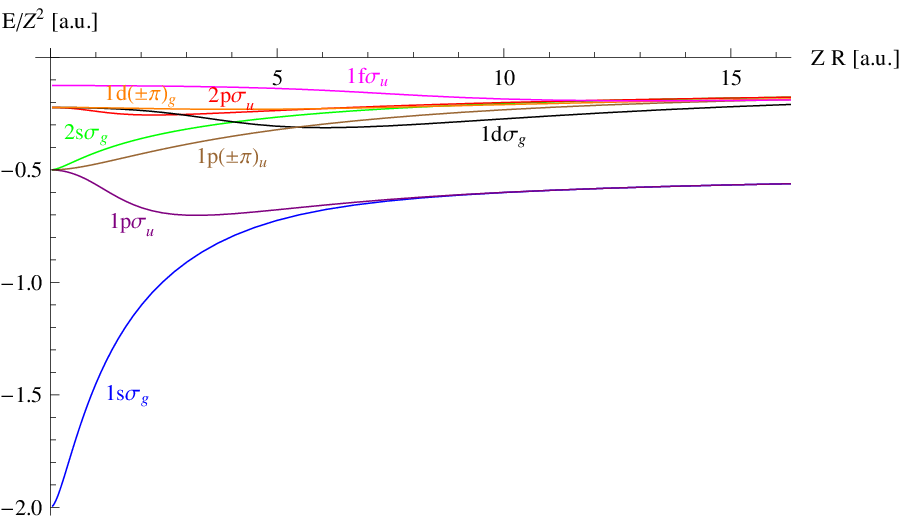}}\\
\subfloat{
\includegraphics[width=0.9\textwidth]{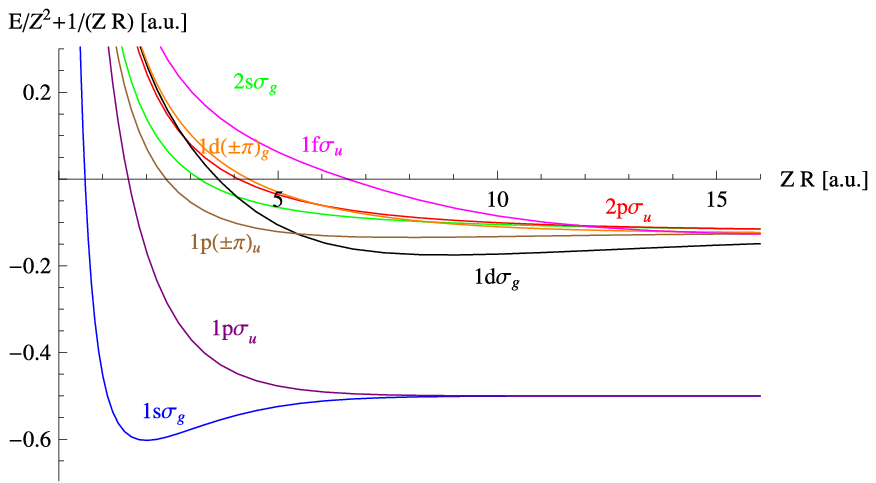}}
\caption{Single-electron energy levels (equation~\eqref{eq:EnergySingle}) of a $\mathrm{H}_2^+$-like homonuclear dimer with respect to $Z R$ (nuclear charge $\times$ nuclear distance), in atomic units. The bottom plot additionally includes the rescaled nuclear-nuclear repulsion term $1/(Z R)$. The unified-atom limit $R \to 0$ corresponds to a hydrogen-like atomic ion with one electron, nuclear charge $2 Z$ and energy levels $-2 Z^2/n^2$, in agreement with the curves of the top subfigure. In the dissociation limit $R \to \infty$, the dimer splits into a single hydrogen-like atom/ion and an isolated nucleus ($\mathrm{H} + \mathrm{p}$ for $Z = 1$). Thus, the electronic energy levels converge to $-\frac12 Z^2/n^2$.}
\label{fig:DiatomicEnergyLevels}
\end{figure}
Figure~\ref{fig:DiatomicEnergyLevels} shows the lowest few homonuclear energy levels in dependence of $Z R$, both with and without the (rescaled) nuclear repulsion term $1/(Z R)$. In analogy to the molecular term symbol, we employ the notation
\begin{equation}
\label{eq:MolecularTermSymbol}
n\ell \phantom{1}^{2s+1}m_{g/u}
\end{equation}
to label states. In common notation, $\ell = 0,1,2,3,\dots$ is designated by $\mathrm{s,p,d,f,\dots}$, respectively, and $m = 0, \pm1, \pm2, \dots$ by $\mathrm{\sigma,\pm\pi,\pm\delta,\dots}$. For fixed $(\ell,m)$, the ``principal value'' $n = 1,2,\dots$ enumerates successive energy levels. In the homonuclear case, the angular spheroidal wave function determines the \emph{parity} $(-1)^\ell$ (reflection about the origin, $\vect{x} \to -\vect{x}$). It is written as \emph{g}erade (even) or \emph{u}ngerade (odd). We omit the spin variable $s$ for now, which will become important for the many-electron calculations in section~\ref{sec:Coulomb}.

Having the exact solution of the two-center Schr\"odinger equation available calls for a comparison with the popular LCAO approach (linear combination of atomic orbitals). Figure~\ref{fig:MolecularOrbitals} tries to match the corresponding wavefunctions, taking parity and ordering of energy levels into account. However, note that the suggestive ordering has to be interpreted with caution since it depends on the nuclear distance $R$. For example, according to figure~\ref{fig:DiatomicEnergyLevels},
\begin{equation*}
E_\mathrm{1s\sigma_g} < E_\mathrm{1p\sigma_u}, E_\mathrm{1p(\pm\pi)_u}, E_\mathrm{2s\sigma_g} < E_\mathrm{2p\sigma_u}, \dots
\end{equation*}
for small nuclear distances $R$. This is different from the arrangement in figure~\ref{fig:MolecularOrbitals}.

\begin{figure}[!ht]
\centering
\includegraphics[width=\textwidth]{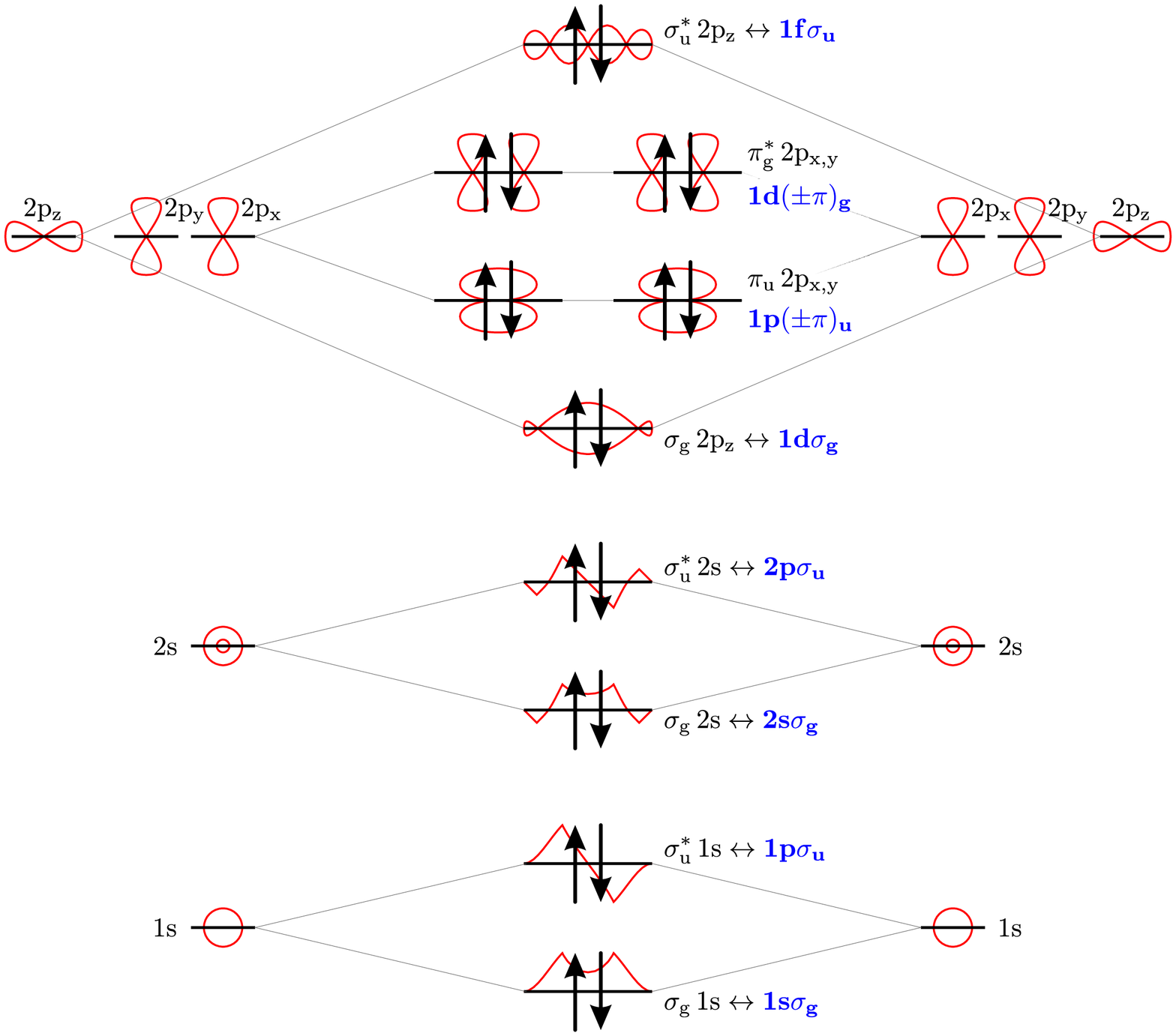}
\caption{Putative best match of the exact $\mathrm{H}_2^+$-like electronic wavefunctions (labeled $n \ell m_{g/u}$ in boldface blue) with the LCAO-MOs (molecular orbitals built from linear combinations of atomic orbitals) widely used in the literature (see e.g.~\cite{Atkins2006}). In particular, the parity (reflection about the origin, $\vect{x} \to -\vect{x}$) agrees in each instance. Orbitals are schematically drawn in red, and antibonding MOs are marked by a star ($^*$).}
\label{fig:MolecularOrbitals}
\end{figure}

\paragraph{Normalization}

In what follows, we derive a formula for the required normalization factor of the wavefunction. The volume element in prolate spheroidal coordinates equals $\ud V = (R/2)^3 \left(\xi^2 - \eta^2\right) \ud\xi \ud\eta \ud\varphi$. Thus
\begin{equation*}
\norm{\psi}_{L^2}^2 = \left(R/2\right)^3 \int_1^\infty \int_{-1}^1 \Lambda(\xi)^2 S(\eta)^2 \left(\xi^2 - \eta^2\right) \ud\eta \ud\xi.
\end{equation*}
The inner integral without the factor $\eta^2$ is already solved in~\eqref{eq:AngularCoeffNormalization}. To include $\eta^2$, we use the identity
\begin{equation}
\label{eq:LegendreMultX}
x \cdot P^\mu_k(x) = \frac{k+\mu}{2k+1} P^\mu_{k-1}(x) + \frac{k-\mu+1}{2k+1} P^\mu_{k+1}(x).
\end{equation}
Thus, after taking into account the normalization factors in the expansion~\eqref{eq:SpheroidalAngularExpansion}, we obtain
\begin{equation*}
\int_{-1}^1 S(\eta)^2\,\eta^2 \ud\eta = \norm{X_\mathrm{Leg}^\mu\vect{c}}^2, \quad \vect{c} \equiv \left(c^\mu_{\ell,k}(p)\right)_k, 
\end{equation*}
with the tridiagonal, symmetric matrix $X_\mathrm{Leg}^\mu$ given by
\begin{equation*}
X_{\mathrm{Leg},kk}^\mu = 0, \quad X_{\mathrm{Leg},k,k+1}^\mu = \left(\frac{(k+1-\mu)(k+1+\mu)}{(2k+1)(2k+3)}\right)^{1/2}, \quad k = \mu, \mu+1, \dots
\end{equation*}

We proceed analogously for the radial part. After a change of variables $x := 2 p\,(\xi-1)$ and due to the orthogonality~\eqref{eq:LaguerreOrthogonality}, we obtain
\begin{equation*}
\int_1^\infty \Lambda(\xi)^2 \ud\xi = \frac{1}{2p} \norm{\vect{d}}^2,
\end{equation*}
where $\vect{d}$ contains the expansion coefficients in~\eqref{eq:HylleraasWavefunction}. To incorporate the factor $\xi^2$, we employ the following well-known relation for Laguerre polynomials:
\begin{equation}
\label{eq:LaguerreMultX}
x \cdot L^\mu_k(x) = -(k+1) L^\mu_{k+1}(x) + (2k+\mu+1)L^\mu_k(x) - (k+\mu)L^\mu_{k-1}(x).
\end{equation}
Thus, multiplying an expansion~\eqref{eq:HylleraasExpansion} by $x$ yields
\begin{equation}
\label{eq:HylleraasMultX}
x \cdot H^\mu_{\vect{d}}(x) = H^\mu_{\vect{d}'}(x), \quad \vect{d}' := X^\mu\vect{d}
\end{equation}
with the tridiagonal, symmetric matrix $X_\mathrm{Lag}^\mu$ defined by
\begin{equation*}
X_{\mathrm{Lag},kk}^\mu = 2k+\mu+1, \quad X_{\mathrm{Lag},k,k+1}^\mu = -\left((k+1)(k+\mu+1)\right)^{1/2}, \quad k = 0,1,\dots
\end{equation*}
Plugging~\eqref{eq:HylleraasMultX} into the following integral yields
\begin{equation*}
\int_1^\infty \Lambda(\xi)^2\, \xi^2\ud\xi = \frac{1}{2p} \norm{\left(I + (2p)^{-1} X_\mathrm{Lag}^\mu\right)\vect{d}}^2
\end{equation*}

Assembling the radial and angular contributions finally results in
\begin{equation*}
\norm{\psi}_{L^2}^2 = \frac{(R/2)^3}{2p} \left(\norm{\vect{c}}^2 \norm{\vect{d} + (2p)^{-1} X_\mathrm{Lag}^\mu \vect{d}}^2 - \norm{X_\mathrm{Leg}^\mu\vect{c}}^2 \norm{\vect{d}}^2\right).
\end{equation*}
That is, we obtain the correct normalization factor directly from the expansion coefficients $\vect{c}$ and $\vect{d}$.

\paragraph{Dissociation limit $R \to \infty$}

From a physical point of view, separating the nuclei from each other should yield a hydrogen-like atom/ion plus an isolated nucleus. However, in the homonuclear case, the symmetry properties of the electronic wavefunctions ($\psi(\vect{x}) = (-1)^\ell\,\psi(-\vect{x})$ due to parity) imply that the electronic charge is equally distributed to both nuclei. This seeming contradiction can be reconciled by constructing superpositions of even and odd wavefunctions to obtain the well-known hydrogen-like wavefunctions, localized at either one or the other nucleus. (Note that conversely, the LCAO approach uses linear combinations of atomic orbitals as molecular wavefunctions.)

From the above arguments, we expect the energy levels to converge to $-\frac12 Z^2/n^2$ in the limit $R \to \infty$, as indicated in figure~\ref{fig:DiatomicEnergyLevels}. Along with it comes a heuristic understanding of the convergence rate\footnote{I am grateful to Gero Friesecke for helpful discussion regarding this point.}. Each ``half'' electron localized at a nucleus experiences an additional attraction from the respective other nucleus. This adds up to the net attraction energy
\begin{equation}
\label{eq:NetAttraction}
-\frac{0.5 \times Z}{R} - \frac{0.5 \times Z}{R} = -\frac{Z}{R}.
\end{equation}
Subtracting this correction term (which of course vanishes as $R \to \infty$) from the energy $E$ leads to exponential (instead of algebraic) convergence, as shown in figure~\ref{fig:DissociationEnergyConvergence}. Namely, the electronic charge distributions decay exponentially with distance from the nuclei, implying a likewise decay of the error.
\begin{figure}[!ht]
\centering
\includegraphics[width=0.8\textwidth]{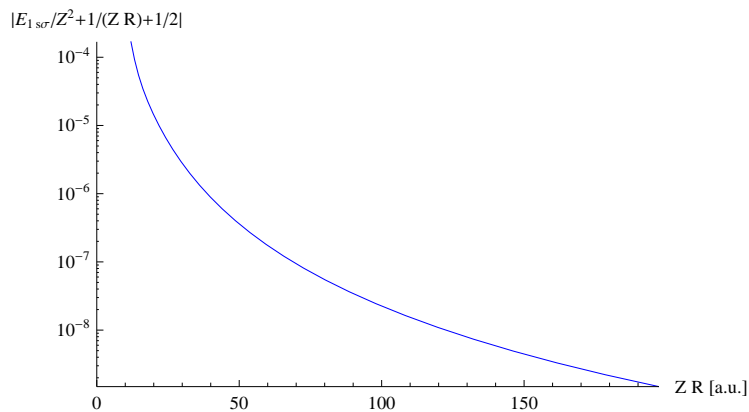}
\caption{Exponential convergence of the $\mathrm{1s\sigma}$ energy level $E/Z^2$ plus the $\frac{1}{Z R}$ correction term~\eqref{eq:NetAttraction} to the hydrogen groundstate energy $-\frac{1}{2}$, as $R \to \infty$.}
\label{fig:DissociationEnergyConvergence}
\end{figure}

\section{Properties of Laguerre Expansions}
\label{sec:LaguerreExpansions}

This technical section is based on function expansions in terms of associated Laguerre polynomials (see equation~\eqref{eq:HylleraasExpansion}). We develop a computational framework for multiplying these expansions, and derive analytic solutions of integrals appearing in section~\ref{sec:Coulomb}.

\subsection{Products of Laguerre Expansions}
\label{sec:ProductLaguerre}

We want to solve the following task: Given integers $m_1, m_2 \in \Z$ and exponentially decaying sequences $(d_{1,k}), (d_{2,k})$, calculate the sequence $(d_k)$ satisfying
\begin{equation}
\label{eq:HylleraasMult}
H^{\abs{m_1}}_{\vect{d}_1}(x) \cdot H^{\abs{m_2}}_{\vect{d}_2}(x) \stackrel{!}{=} H^{\abs{m_1-m_2}}_{\vect{d}}(x).
\end{equation}
For conciseness of notation, let $\mu_i := \abs{m_i}$, $i = 1,2$ and $\mu_3 := \abs{m_1-m_2}$, and assume without loss of generality that $\mu_1 \ge \mu_2$. Depending on the signs of $m_1$ and $m_2$, we have $\mu_3 = \mu_1 \pm \mu_2$. The orthogonality relation of Laguerre polynomials leads to
\begin{equation}
\label{eq:HylleraasMultCoeff}
d_k = \hprod{\vect{d}_2}{\Pi^{\vect{\mu}}_k\,\vect{d}_1}, \quad k = 0,1,2,\dots
\end{equation}
with the symmetric matrix $\Pi^{\vect{\mu}}_k \equiv (a^{\vect{\mu}}_{ijk})_{ij}$ given by
\begin{equation}
\label{eq:a_m_ijk}
\begin{split}
a^{\vect{\mu}}_\vect{i}
&:= \int_0^\infty H^{\mu_1}_{i_1}(x) H^{\mu_2}_{i_2}(x) H^{\mu_3}_{i_3}(x) \ud x\\
&= \left(\prod_{k=1}^3 \frac{i_k!}{(i_k+\mu_k)!}\right)^{1/2} \times
\begin{cases}
\,b^\vect{\mu}_\vect{i}(3/2)& \text{if $\mu_3 = \mu_1-\mu_2$}\\
\,\tilde{b}^\vect{\mu}_\vect{i}(3/2)& \text{if $\mu_3 = \mu_1+\mu_2$}
\end{cases}
\end{split}
\end{equation}
In the above expression,
\begin{align}
\label{eq:bmi_diff}
b^\vect{\mu}_\vect{i}(z)
&:= \int_0^\infty x^{\mu_1} L^{\mu_1}_{i_1}(x) L^{\mu_2}_{i_2}(x) L^{\mu_3}_{i_3}(x)\,\mathrm{e}^{-z\,x} \ud x,\\
\label{eq:bmi_sum}
\tilde{b}^\vect{\mu}_\vect{i}(z)
&:= \int_0^\infty x^{\mu_1+\mu_2} L^{\mu_1}_{i_1}(x) L^{\mu_2}_{i_2}(x) L^{\mu_3}_{i_3}(x)\,\mathrm{e}^{-z\,x} \ud x
\end{align}
defined for $z \in \R_{>0}$. Using the recurrence relation
\begin{equation*}
x^\mu L^\mu_i(x) = (i+\mu) x^{\mu-1} L^{\mu-1}_i(x) - (i+1) x^{\mu-1} L^{\mu-1}_{i+1}(x),
\end{equation*}
the integrals~\eqref{eq:bmi_diff} and~\eqref{eq:bmi_sum} can be reduced to the following proposition, which is a generalization of~\cite{LaguerreProducts1960}.

\begin{proposition}
Given fixed integers $\vect{\mu} \in \N_0^3$, the coefficients
\begin{equation}
\label{eq:cmi_z}
c^\vect{\mu}_\vect{i}(z) := \int_0^\infty L^{\mu_1}_{i_1}(x) L^{\mu_2}_{i_2}(x) L^{\mu_3}_{i_3}(x)\,\mathrm{e}^{-z\,x} \ud x, \quad z \in \R_{>0}
\end{equation}
defined for $\vect{i} \in \N_0^3$ obey the recurrence relation
\begin{equation}
\label{eq:cmi_z_rec}
\begin{split}
c^\vect{\mu}_\vect{i}(z) =
&- (1/z-1) \left(c^\vect{\mu}_{i_1-1,i_2,i_3}(z) + c^\vect{\mu}_{i_1,i_2-1,i_3}(z) + c^\vect{\mu}_{i_1,i_2,i_3-1}(z)\right)\\
&+ (2/z-1) \left(c^\vect{\mu}_{i_1,i_2-1,i_3-1}(z) + c^\vect{\mu}_{i_1-1,i_2,i_3-1}(z) + c^\vect{\mu}_{i_1-1,i_2-1,i_3}(z)\right)\\
&- (3/z-1) c^\vect{\mu}_{i_1-1,i_2-1,i_3-1}(z)\\
&+ \frac{1}{z} \prod_{k=1}^3 \binom{\mu_k-1+i_k}{i_k}
\end{split}
\end{equation}
with the convention that $c^\vect{\mu}_\vect{i}(z) = 0$ if any $i_1,i_2,i_3 < 0$ and $\binom{i-1}{i} = \delta_{0i}$ for integer $i \ge 0$.
\end{proposition}

\noindent Thus, $c^\vect{\mu}_\vect{i}(z)$ can iteratively be calculated and stored for later usage. Note that the coefficients $c^\vect{0}_\vect{i}(z)$ are symmetric in $i_1,i_2,i_3$. The case $z = 1$ and $\vect{\mu} = \vect{0}$ is handled in~\cite{LaguerreProducts1960} (with a sign typo in his equation~(7)). From the particular form of the binomial coefficients in~\eqref{eq:cmi_z_rec} it follows that the recurrence relation is homogeneous precisely if any $\mu_k = 0$.

\begin{proof}
A derivation of~\eqref{eq:cmi_z_rec} proceeds along the same lines as in~\cite{LaguerreProducts1960}, involving generating functions of Laguerre polynomials. More specifically, using
\begin{equation*}
\sum_{i=0}^\infty L^\mu_i(x) (-t)^i = \frac{\mathrm{e}^{x\,t/(1+t)}}{(1+t)^{\mu+1}},
\end{equation*}
the following formal series in $t_1, t_2, t_3$ fulfills
\begin{equation}
\label{eq:gen_m}
G^\vect{\mu}(t_1,t_2,t_3,z) := \sum_{i_1,i_2,i_3=0}^\infty c^\vect{\mu}_\vect{i}(z) \prod_{k=1}^3 (-t_k)^{i_k} = z^{-1} \frac{\prod_{k=1}^3(1+t_k)^{-\mu_k}}{1-p(t_1,t_2,t_3,z)}
\end{equation}
with
\begin{equation*}
\begin{split}
&p(t_1,t_2,t_3,z) := (1/z-1)\left(t_1 + t_2 + t_3\right)\\
&\quad + (2/z-1)\left(t_1 t_2 + t_1 t_3 + t_2 t_3\right) + (3/z-1) t_1 t_2 t_3.
\end{split}
\end{equation*}
Applying the identity $1/(1-x) = 1 + x/(1-x)$ for $x = p(t_1,t_2,t_3,z)$ to the right hand side of~\eqref{eq:gen_m} leads to
\begin{equation*}
G^\vect{\mu}(t_1,t_2,t_3,z) = p(t_1,t_2,t_3,z)\,G^\vect{\mu}(t_1,t_2,t_3,z) + \frac{1}{z} \prod_{k=1}^3(1+t_k)^{-\mu_k}.
\end{equation*}
Now comparing coefficients of $t_1^{i_1}\,t_2^{i_2}\,t_3^{i_3}$ gives equation~\eqref{eq:cmi_z_rec}.
\end{proof}

\begin{figure}[!ht]
\centering
\includegraphics[width=0.8\textwidth]{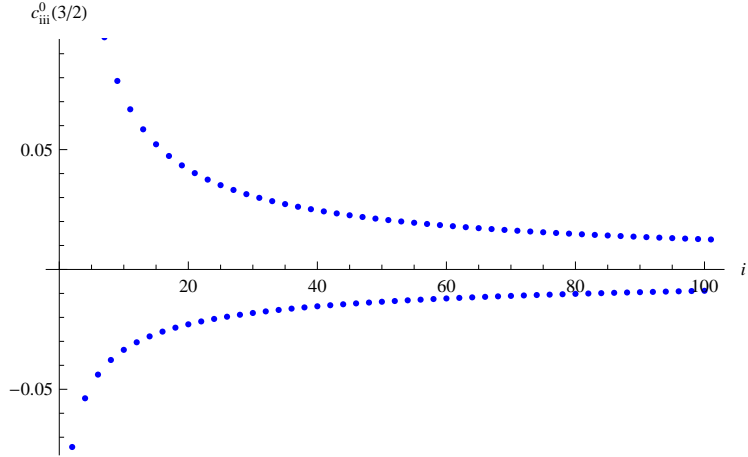}
\caption{Asymptotic behavior of the central coefficient $c^\vect{0}_{i,i,i}(3/2)$ defined in~\eqref{eq:cmi_z}, which oscillates between positive and negative values.}
\label{fig:LaguerreProdCoeff}
\end{figure}
Numeric experimentation suggests that $c^\vect{0}_\vect{i}(z)$ is bounded asymptotically ($\abs{\vect{i}} \to \infty$) if and only if $z \ge 3/2$. As illustration, figure~\ref{fig:LaguerreProdCoeff} shows the central coefficient $c^\vect{0}_{i,i,i}(3/2)$, which alternates its sign depending on the parity of $i$. As heuristic explanation of the asymptotic behavior, we focus on the central coefficient $c^\vect{0}_{i,i,i}(z)$ and set
\begin{equation*}
\tilde{c}_{3i}(z) := c^\vect{0}_{i,i,i}(z), \quad
\tilde{c}_{3i-1}(z) := c^\vect{0}_{i-1,i,i}(z), \quad
\tilde{c}_{3i-2}(z) := c^\vect{0}_{i-1,i-1,i}(z).
\end{equation*}
Plugged into~\eqref{eq:cmi_z_rec} and letting $k := 3 i$ gives
\begin{equation*}
\tilde{c}_k(z) = - 3\,(1/z-1)\,\tilde{c}_{k-1}(z) + 3\,(2/z-1)\,\tilde{c}_{k-2}(z) - (3/z-1)\,\tilde{c}_{k-3}(z).
\end{equation*}
This equation is only correct if $k$ is a multiple of $3$. Nevertheless, interpreted as difference equation yields the companion matrix
\begin{equation*}
Z = \begin{pmatrix}0&1&0\\0&0&1\\-(3/z-1)&3(2/z-1)&-3(1/z-1)\end{pmatrix}
\end{equation*}
with eigenvalues $\{1, 1, -(3/z-1)\}$. Thus, the spectral radius $\rho(Z) \le 1$ precisely if $z \ge 3/2$.

We add the following observation: the homogeneous recurrence relation may be interpreted as a differential equation in 3 dimensions by treating the indices $\vect{i}$ as continuous variables, $c^\vect{\mu}_\vect{i}(z) \equiv f^\vect{\mu}_z(\vect{i})$, and taking the continuity limit. Namely, without the inhomogeneous contribution,~\eqref{eq:cmi_z_rec} becomes
\begin{equation*}
\begin{split}
0 &\stackrel{!}{=} \frac{1}{h^2}\Big[- (1/z-1) \left(f^\vect{\mu}_z(\vect{i}-(h,0,0)) + f^\vect{\mu}_z(\vect{i}-(0,h,0)) + f^\vect{\mu}_z(\vect{i}-(0,0,h))\right)\\
&\qquad\quad + (2/z-1) \left(f^\vect{\mu}_z(\vect{i}-(0,h,h)) + f^\vect{\mu}_z(\vect{i}-(h,0,h)) + f^\vect{\mu}_z(\vect{i}-(h,h,0))\right)\\
&\qquad\quad - (3/z-1) f^\vect{\mu}_z(\vect{i}-(h,h,h)) - f^\vect{\mu}_z(\vect{i})\Big]\\
&= -\frac{1}{z}\left(\partial_{i_2 i_3} + \partial_{i_1 i_3} + \partial_{i_1 i_2}\right)f^\vect{\mu}_z(\vect{i}) + h \left(\frac{3}{2z}-1\right)\partial_{i_1 i_2 i_3}\,f^\vect{\mu}_z(\vect{i}) + \mathcal{O}\left(h^2\right).
\end{split}
\end{equation*}
Here we have already used
\begin{equation*}
\left(\partial_{i_2 i_3} + \partial_{i_1 i_3} + \partial_{i_1 i_2}\right)f^\vect{\mu}_z(\vect{i}) = 0
\end{equation*}
to simplify the $\mathcal{O}(h)$ term, which disappears precisely for $z = 3/2$.

\subsection{Argument Rescaling}
\label{sec:ArgRescale}

Given any fixed $y \in \R_{>0}$, we try to re-express Laguerre expansions~\eqref{eq:HylleraasExpansion} evaluated at the scaled coordinates $y\,x$ as expansions evaluated at $x$. First note the following well-known identity for $k, \mu \in \N_0$,
\begin{equation}
\label{eq:LaguerreArgscale}
L^\mu_k\left(y\,x\right) = y^k \sum_{i=0}^k (1/y-1)^{k-i} \binom{k+\mu}{i+\mu} L^\mu_i(x) \quad \text{for} \quad y \in \R_{>0}.
\end{equation}
Similarly, a direct calculation shows that for all $y \neq 0$,
\begin{equation}
\label{eq:LaguerreArgscaleShifted}
\sum_{i=0}^k (1-y)^{k-i} L^\mu_i(y\,x) = y^k \sum_{i=0}^k (1/y-1)^{k-i} \binom{k+\mu+1}{i+\mu+1} L^\mu_i(x).
\end{equation}
Due to~\eqref{eq:LaguerreArgscale}, for any exponentially decaying sequence $\vect{d} := (d_k)_{k \ge 0}$ it holds that
\begin{equation}
\label{eq:LaguerreNormArgscale}
H^\mu_\vect{d}(y\,x) = y^{\mu/2} \mathrm{e}^{-(y-1)x/2}\,H^\mu_{\vect{d}'}(x), \quad \vect{d}' := S^\mu_y\,\vect{d}
\end{equation}
with the upper triangular matrix $S^\mu_y \equiv (s^\mu_{ik}(y))$ defined by
\begin{equation*}
s^\mu_{ik}(y) := \left(\frac{k!}{i!}\frac{(i+\mu)!}{(k+\mu)!}\right)^{1/2} y^k (1/y-1)^{k-i} \binom{k+\mu}{i+\mu} \quad \text{for} \quad i \le k
\end{equation*}
and $s^\mu_{ik}(y) = 0$ otherwise.

This result can be combined with the operation~\eqref{eq:HylleraasMult}, as follows. Assume we are given $\vect{\mu} \in \N_0^3$ with $\mu_3 = \mu_1 \pm \mu_2$, as well as $z_1, z_2 \in \R_{>0}$ and two exponentially decaying sequences $\vect{d}_1, \vect{d}_2$. Set $z := (z_1 + z_2)/2$ and use~\eqref{eq:HylleraasMultCoeff} to calculate the sequence $\vect{d} \equiv (d_k)_{k \ge 0}$,
\begin{equation*}
d_k := \left(z_1'\right)^{\mu_1/2} \left(z_2'\right)^{\mu_2/2} \hprod{S^{\mu_2}_{z_2'}\,\vect{d}_2}{\Pi^{\vect{\mu}}_k\, S^{\mu_1}_{z_1'}\,\vect{d}_1}, \quad z_i' := z_i/z.
\end{equation*}
Then, combining~\eqref{eq:HylleraasMult} with~\eqref{eq:LaguerreNormArgscale} gives
\begin{equation}
\label{eq:HylleraasMultArgscale}
H^{\mu_1}_{\vect{d}_1}(z_1\,x)\, H^{\mu_2}_{\vect{d}_2}(z_2\,x) = H^{\mu_3}_{\vect{d}}(z\,x).
\end{equation}
Note that the exponential functions on both sides match. Summarizing, we have obtained the product of two Laguerre expansions with rescaled arguments.

As slight variation of~\eqref{eq:LaguerreNormArgscale}, given $y \in \R_{>0}$ and an exponentially decaying sequence $\vect{d}$, we try to find a sequence $\vect{d}'$ such that
\begin{equation}
\label{eq:HylleraasArgscale}
H^\mu_\vect{d}(y\,x) \stackrel{!}{=} H^\mu_{\vect{d}'}(x).
\end{equation}
Due to the orthogonality relation~\eqref{eq:LaguerreOrthogonality}, we have to compute the following integral for integers $i, k \ge 0$. Using equations~\eqref{eq:LaguerreArgscale} and~\eqref{eq:LaguerreOrthogonality} leads to
\begin{multline*}
\int_0^\infty H^\mu_i(x) H^\mu_k(y\,x) \ud x
= \frac{2}{y+1} \left(\frac{2 \sqrt{y}}{y+1}\right)^\mu \binom{i+\mu}{\mu}^{1/2}\binom{k+\mu}{\mu}^{1/2}\\
\times (-1)^k \left(\frac{y-1}{y+1}\right)^{i+k} {_2F_1}\left(-i,-k;1+\mu;-\frac{4y}{(y-1)^2}\right)
\end{multline*}
with the Gaussian hypergeometric function ${_2F_1}$.

\subsection{Integral Identities}
\label{sec:LaguerreExpArcothIntegral}

We derive analytic solutions of integrals originating from Neumann's expansion of $1/\abs{\vect{x}-\vect{y}}$ in terms of Legendre functions (see equation~\eqref{eq:NeumannExpansion} below).

\begin{proposition}
\label{prop:LaguerreExpZInt}
For any $y, z \in \R_{>0}$ and integers $k, \mu \ge 0$, it holds that
\begin{equation}
\label{eq:LaguerreExpZInt}
\begin{split}
&\int_0^z L^\mu_k(y\,x)\,\mathrm{e}^{-x} \ud x = (1-y)^k - L^\mu_k(y\,z)\,\mathrm{e}^{-z}\\
&\quad + \sum_{i=0}^{k-1} (1-y)^{k-1-i} \left(y\,L^\mu_i(y\,z)\,\mathrm{e}^{-z} + \binom{i+\mu}{i+1}\right).
\end{split}
\end{equation}
\end{proposition}
\begin{proof}
First note that for $y = 1$, equation~\eqref{eq:LaguerreExpZInt} simplifies to $1 - \mathrm{e}^{-z}$ for $k = 0$ and
\begin{equation}
\label{eq:LaguerreExp1Int}
\int_0^z L^\mu_k(x)\,\mathrm{e}^{-x} \ud x = \left(L^\mu_{k-1}(z) - L^\mu_k(z)\right)\mathrm{e}^{-z} + \binom{k+\mu-1}{k} \quad \text{for} \quad k \ge 1.
\end{equation}
This identity can be proven by taking derivatives on both sides. Then, combining~\eqref{eq:LaguerreArgscale} and~\eqref{eq:LaguerreArgscaleShifted} (for $k-1$) with~\eqref{eq:LaguerreExp1Int} leads to~\eqref{eq:LaguerreExpZInt}.
\end{proof}

For any integers $\mu, k, \tilde{k} \ge 0$, consider the nested integrals
\begin{equation}
\label{eq:LaguerreNestedExpInt}
t^\mu_{k \tilde{k}} := \frac{1}{2} \int_0^\infty L^\mu_{\tilde{k}}(\tilde{x})\,\mathrm{e}^{-\tilde{x}/2} \int_0^{\tilde{x}} L^\mu_k(x)\,\mathrm{e}^{-x/2} \ud x \ud \tilde{x}.
\end{equation}
They have a surprisingly simple form for $\mu = 0, 1$:
\begin{proposition}
\label{prop:LaguerreNestedExpInt}
The integrals $t^0_{k \tilde{k}}$ in~\eqref{eq:LaguerreNestedExpInt} are equal to
\begin{equation*}
t^0_{k\tilde{k}} = \begin{cases}
1& k = \tilde{k}\\
2\,(-1)^{k+\tilde{k}}& k < \tilde{k}\\
0& k > \tilde{k}
\end{cases}
\qquad\qquad\quad
\left(t^0_{k\tilde{k}}\right) =
\left(\begin{smallmatrix}
1& -2& \phantom{-}2& -2&\\
0& \phantom{-}1& -2& \phantom{-}2&\\
0& \phantom{-}0& \phantom{-}1& -2&\cdots\\
0& \phantom{-}0& \phantom{-}0& \phantom{-}1&\\
&&\vdots&&
\end{smallmatrix}\right),
\end{equation*}
and the integrals $t^1_{k \tilde{k}}$
\begin{equation*}
t^1_{k\tilde{k}} = \begin{cases}
(-1)^{\tilde{k}}& \text{$k \le \tilde{k}$ and $k$ even}\\
1& \text{$k > \tilde{k}$ and $\tilde{k}$ even}\\
0& \text{otherwise}
\end{cases}
\qquad
\left(t^1_{k\tilde{k}}\right) =
\left(\begin{smallmatrix}
1& -1& \phantom{-}1& -1&\\
1& \phantom{-}0& \phantom{-}0& \phantom{-}0&\\
1& \phantom{-}0& \phantom{-}1& -1&\cdots\\
1& \phantom{-}0& \phantom{-}1& \phantom{-}0&\\
&&\vdots&&
\end{smallmatrix}\right).
\end{equation*}
\end{proposition}
\begin{proof}
These identities can be proven by applying proposition~\ref{prop:LaguerreExpZInt} to the inner integral and using the orthogonality property of the Laguerre polynomials.
\end{proof}

For the following paragraph, we state
\begin{definition}
Given integers $\mu \in \N_0$ and $1 \le i \le k$, set
\begin{equation}
\label{eq:hki}
h^\mu_{ki}(y) := \sum_{n=i}^k (-y)^{n-i} \binom{k+\mu}{n+\mu}\Big/\binom{n}{i}, \quad y \in \R.
\end{equation}
Expressed in terms of generalized hypergeometric functions,
\begin{equation*}
h^\mu_{ki}(y) = \binom{k+\mu}{i+\mu}\, {_3F_2}\begin{pmatrix}1&1&i-k&\multirow{2}{*}{; y}\\1+i&1+i+\mu\end{pmatrix}.
\end{equation*}
\end{definition}

Given $z \in \R_{>0}$, $y \in \R$ and integers $k, \mu \ge 0$, we set out to solve the integral
\begin{equation}
\label{eq:LaguerreArcothInt}
\int_0^\infty L^\mu_k(y\,x)\,\mathrm{e}^{-x}\,\log\left(1+\frac{z}{x}\right) \ud x.
\end{equation}
For that purpose, we decompose the logarithm into $\log(1+x/z) - \log(x/z)$. Considering the first term, integration by parts and~\eqref{eq:LaguerreExpZInt} give
\begin{equation*}
\begin{split}
&\int_0^\infty \log(1+x/z)\, L^\mu_k(y\,x)\,\mathrm{e}^{-x} \ud x
= \int_0^\infty \frac{1}{x+z} L^\mu_k(y\,x) \ud x\\
&\quad - y \sum_{i=0}^{k-1} (1-y)^{k-1-i} \int_0^\infty \frac{1}{x+z} L^\mu_i(y\,x)\mathrm{e}^{-x} \ud x.
\end{split}
\end{equation*}
The integrals on the right hand side are solved by the following proposition:

\begin{proposition}
\label{prop:InvXZLaguerreInt}
Let $z \in \R_{>0}$ and $y \in \R$, then for all integers $k, \mu \ge 0$,
\begin{equation}
\label{eq:InvXZLaguerreInt}
\int_0^\infty \frac{1}{x+z}\, L^\mu_k(y\,x)\,\mathrm{e}^{-x} \ud x = L^\mu_k(-y\,z)\,\Gamma(0,z)\,\mathrm{e}^{z} - \frac{1}{z} \sum_{i=1}^k h^\mu_{ki}(y) \frac{(y\,z)^i}{i!}
\end{equation}
with the incomplete gamma function $\Gamma$.
\end{proposition}
\begin{proof}
For $k = 0$, we obtain (via a computer algebra system)
\begin{equation*}
\int_0^\infty \frac{1}{x+z}\,\mathrm{e}^{-x} \ud x = \Gamma(0,z)\,\mathrm{e}^{z}
\end{equation*}
in agreement with the right hand side of~\eqref{eq:InvXZLaguerreInt}. For $k \ge 1$, a change of variables yields
\begin{equation*}
\int_0^\infty \frac{1}{x+z}\,(y\,x)^i\,\mathrm{e}^{-x} \ud x = \int_0^\infty \frac{1}{x+1}\,(y\,z\,x)^i\,\mathrm{e}^{-z\,x} \ud x = (-y\,z)^i \frac{\ud^i}{\ud z^i} \Gamma(0,z)\,\mathrm{e}^{z}
\end{equation*}
for any integer $i \ge 0$. Thus, the integral~\eqref{eq:InvXZLaguerreInt} is a linear combination of the last term ($i = 0,\dots,k$). The explicit formula on the right hand side of~\eqref{eq:InvXZLaguerreInt} follows from a (rather tedious) calculation, using $\ud_z\,\Gamma(0,z) = -\mathrm{e}^{-z}/z$.
\end{proof}

Concerning the second logarithm $\log(x/z)$ in the above decomposition, first note that \begin{equation*}
f_i(z) := \int_0^\infty \log(x/z)\,\frac{x^i}{i!}\,\mathrm{e}^{-x} \ud x = H_i - (\gamma + \log(z)),
\end{equation*}
where $\gamma$ is Euler's constant and $H_i$ the $i^\mathrm{th}$ Harmonic number. Namely, integration by parts yields the recurrence relation
\begin{equation*}
f_i(z) = \frac{1}{i} + f_{i-1}(z), \quad i = 1,2,\dots,
\end{equation*}
and $f_0(z) = -(\gamma + \log(z))$ can be shown by a computer algebra system. Thus, for all integers $k, \mu \ge 0$,
\begin{equation*}
\int_0^\infty \log(x/z) L^\mu_k(x)\,\mathrm{e}^{-x} \ud x = \sum_{i=1}^k \binom{k+\mu}{i+\mu} (-1)^i H_i - \binom{k+\mu-1}{k} (\gamma + \log(z)).
\end{equation*}
Combining this equation with~\eqref{eq:LaguerreArgscale} yields the following generalization:
\begin{proposition}
Let $z \in \R_{>0}$ and $y \in \R$, then for all integers $k, \mu \ge 0$,
\begin{multline*}
\int_0^\infty \log\left(x/z\right) L^\mu_k(y\,x)\,\mathrm{e}^{-x} \ud x = \sum_{i=1}^k \binom{k+\mu}{i+\mu} (-y)^i H_i\\
- \binom{k+\mu}{k} \left((1-y)^k + \mu \sum_{i=1}^k \binom{k}{i} \frac{y^i (1-y)^{k-i}}{i+\mu}\right) \left(\gamma + \log(z)\right).
\end{multline*}
\end{proposition}

Hence we have collected all ingredients for solving the integral~\eqref{eq:LaguerreArcothInt} in closed form.

We can now assemble the above results to calculate the nested integral
\begin{equation}
\label{eq:LaguerreNestedArcothInt}
\int_0^\infty L^\mu_{\tilde{k}}(\tilde{x})\,\mathrm{e}^{-\tilde{x}/2}\,\mathrm{arcoth}(1+\tilde{x}/z) \int_0^{\tilde{x}} L^\mu_k(x)\,\mathrm{e}^{-x/2} \ud x \ud \tilde{x}, \quad z \in \R_{>0}
\end{equation}
for integers $k, \tilde{k}, \mu \in \N_0$. Proposition~\ref{prop:LaguerreExpZInt} with $y = 2$ and a change of variables ($2x \to x$) gives the inner integral. Combined with the Laguerre product coefficients $c^\vect{\mu}_\vect{i}(1)$ in~\eqref{eq:cmi_z},
\begin{equation*}
\begin{split}
\eqref{eq:LaguerreNestedArcothInt}
&= 2\,q^\mu_{1,k} \int_0^\infty L^\mu_{\tilde{k}}(2x)\, \mathrm{e}^{-x} \log\left(1+\frac{z}{x}\right) \ud x\\
&\, -\sum_{j=0}^{k+\tilde{k}} q^\mu_{2,k\tilde{k}j} \int_0^\infty L^0_j(x)\, \mathrm{e}^{-x} \log\left(1+\frac{2z}{x}\right) \ud x
\end{split}
\end{equation*}
with the integer (!) coefficients
\begin{align*}
q^\mu_{1,k} &:= (-1)^k + \sum_{i=0}^{k-1} (-1)^{k-1-i} \binom{i+\mu}{i+1},\\
q^\mu_{2,k\tilde{k}j} &:= c^{\mu \mu 0}_{k \tilde{k} j}(1) + 2 \sum_{i=j-\tilde{k}}^{k-1} (-1)^{k-i}\,c^{\mu \mu 0}_{i \tilde{k} j}(1).
\end{align*}
The two above integrals are precisely of the form~\eqref{eq:LaguerreArcothInt}, which completes the calculation of~\eqref{eq:LaguerreNestedArcothInt}.

\medskip

As last task of this section, given $k, \mu \in \N_0$ and $z \in \R_{>0}$, we try to compute the Laguerre expansion coefficients of
\begin{equation*}
\frac{1}{\sqrt{2z+x}} H^\mu_k(x).
\end{equation*}
In other words, due to the orthogonality property of Laguerre polynomials, we have to calculate the integrals
\begin{equation}
\label{eq:SqrtHylleraasInt}
\int_0^\infty \frac{1}{\sqrt{2z+x}} H^\mu_k(x) H^\mu_{\tilde{k}}(x) \ud x
\end{equation}
for $\tilde{k} = 0,1,\dots$. Employing the Laguerre product coefficients $b^{\mu\mu 0}_{k\tilde{k}i}(1)$ from~\eqref{eq:bmi_diff}, the above integral can be reduced to a linear combination of
\begin{equation}
\label{eq:w_k_z}
w_k(z) := \int_0^\infty \frac{1}{\sqrt{2z+x}} L_k(x)\,\mathrm{e}^{-x} \ud x, \quad z \in \R_{>0}.
\end{equation}
We can calculate these integrals iteratively for $k = 0,1,2,\dots$ via the following

\begin{proposition}
The functions $w_k(z)$ obey the recurrence relation
\begin{equation}
\label{eq:w_k_z_recurrence}
w_k(z) = w_{k-1}(z) + \frac{z \sqrt{z}}{k} \frac{\ud}{\ud z}\frac{w_{k-1}(z)}{\sqrt{z}}, \quad k = 1,2,\dots
\end{equation}
with the starting value
\begin{equation*}
w_0(z) = \sqrt{\pi}\,\mathrm{erfc}\left(\sqrt{2\,z}\right) \mathrm{e}^{2\,z},
\end{equation*}
where $\mathrm{erfc}$ is the complementary error function.
\end{proposition}
\begin{proof}
The formula for $w_0(z)$ can be derived via a computer algebra system. Concerning the recurrence relation, first integrate~\eqref{eq:w_k_z} by parts ($x \to x/z$) to obtain the alternative representation
\begin{equation*}
w_k(z) = \sqrt{z} \int_0^\infty \frac{1}{\sqrt{2+x}} L_k(z\,x)\, \mathrm{e}^{-z\,x} \ud x.
\end{equation*}
Applying the relation
\begin{equation}
\label{eq:LaguerreExpDiff}
\frac{\ud}{\ud z} L_k(z\,x)\, \mathrm{e}^{-z\,x} = \frac{k+1}{z} \left(L_{k+1}(z\,x) - L_k(z\,x)\right) \mathrm{e}^{-z\,x}
\end{equation}
to this representation gives the recurrence formula~\eqref{eq:w_k_z_recurrence}. The relation~\eqref{eq:LaguerreExpDiff} follows from combining
\begin{align*}
\frac{\ud}{\ud z} L_k(z)\, \mathrm{e}^{-z} &= -\sum_{i=0}^k L_i(z)\, \mathrm{e}^{-z} \quad \text{with}\\
x \cdot L_i(x) &= -(i + 1) L_{i+1}(x) + (2i+1) L_i(x) - i\,L_{i-1}(x).
\end{align*}
\end{proof}

\section{Coulomb and Exchange Integrals of Prolate Spheroidal Orbitals}
\label{sec:Coulomb}

The computation of Coulomb interactions is often the most demanding task concerning multi-electron quantum systems. In this section, we provide the details of an efficient algorithmic implementation, which employs analytically precomputed integrals (from section~\ref{sec:LaguerreExpansions}) and a subsequent Taylor expansion to speed up calculations, and avoids difficulties caused by an alternative numeric approach. For example, we observe that the absolute value of the nested integrals in equation~\eqref{eq:RadialCoulombInt} is typically much smaller than (the maximum over $x$) of the inner integral. This general effect could be explained by the orthogonality property of Laguerre polynomials. In any case, analytically solving the nested integrals as a whole circumvents the numeric difficulties caused by the blow-up of the inner integral.

Given square-integrable spatial ``orbitals'' $a,b,c,d \in L^2(\R^3,\C)$, we define the \emph{Coulomb integral} (following standard notation) as
\begin{equation*}
\bprod{a b}{c d} := \int_{\R^6} \conj{a(\vect{x}_1)} b(\vect{x}_1)\,\frac{1}{\abs{\vect{x}_1-\vect{x}_2}}\,\conj{c(\vect{x}_2)} d(\vect{x}_2) \ud \vect{x}_1 \vect{x}_2,
\end{equation*}
where $\conj{\,\cdot\,}$ is the complex conjugation.
In our setting, we want to calculate the concrete realization
\begin{equation}
\label{eq:CoulombInt}
\bprod{\psi_{n \ell m} \psi_{n'\ell'm'}}{\psi_{\tilde{n}\tilde{\ell}\tilde{m}} \psi_{\tilde{n}'\tilde{\ell}'\tilde{m}'}}
\end{equation}
for single-electron wavefunctions $\psi_{n \ell m}$ from section~\ref{sec:Framework}. The labels $n \ell m$, $n'\ell'm'$ etc are the ``quantum numbers'' in the molecular term symbol~\eqref{eq:MolecularTermSymbol}.

To evaluate these Coulomb integrals in prolate spheroidal coordinates, we pursue the same approach as~\cite{AubertBessis1974} and employ Neumann's expansion
\begin{equation}
\label{eq:NeumannExpansion}
\begin{split}
\frac{1}{\abs{\vect{x}_1-\vect{x}_2}}
&= \frac{4}{R} \sum_{\tau=0}^{\infty} \sum_{\nu=0}^\tau (-1)^\nu \epsilon_\nu \frac{2\tau+1}{2} \left(\frac{(\tau-\nu)!}{(\tau+\nu)!}\right)^2 P^\nu_\tau(\xi_1) Q^\nu_\tau(\xi_2)\\
&\qquad \times P^\nu_\tau(\eta_1) P^\nu_\tau(\eta_2) \cos(\nu(\varphi_1-\varphi_2))
\end{split}
\end{equation}
with $\epsilon_0 = 1$, $\epsilon_\nu = 2$ for $\nu > 0$ and $\xi_1 < \xi_2$ (otherwise interchange $\xi_1 \leftrightarrow \xi_2$). $P^\nu_\tau$ and $Q^\nu_\tau$ are the Legendre functions of the first and second kind, respectively. A derivation of~\eqref{eq:NeumannExpansion} can be found in~\cite{Ruedenberg1951}. For the following, remember the volume element in prolate spheroidal coordinates, $\ud V = (R/2)^3 \left(\xi^2 - \eta^2\right) \ud\xi \ud\eta \ud\varphi$.

With~\eqref{eq:NeumannExpansion} plugged into~\eqref{eq:CoulombInt}, the integrals over $\varphi_1$ and $\varphi_2$ result in
\begin{equation}
\label{eq:PhiInt}
\begin{split}
&\frac{1}{(2\pi)^2} \int_0^{2\pi} \int_0^{2\pi} \cos(\nu(\varphi_1-\varphi_2)) \,\mathrm{e}^{-\ii(m-m')\varphi_1 - \ii(\tilde{m}-\tilde{m}')\varphi_2} \ud\varphi_1 \ud\varphi_2\\
&= \begin{cases}1/\epsilon_\nu& \text{if $m-m' = -(\tilde{m}-\tilde{m}')$ and $\nu = \abs{m-m'}$}\\0&\text{otherwise}
\end{cases}
\end{split}
\end{equation}
Thus, the right hand side of~\eqref{eq:NeumannExpansion} effectively contains the sum over $\tau$ only, starting from $\tau = \nu$.

To approximate the infinite sum over $\tau$, we include all terms up to a threshold $\tau_{\max}$. This truncation is justified due to the exponential convergence, as illustrated in figure~\ref{fig:TauConv}.
\begin{figure}[!ht]
\centering
\includegraphics[width=0.85\textwidth]{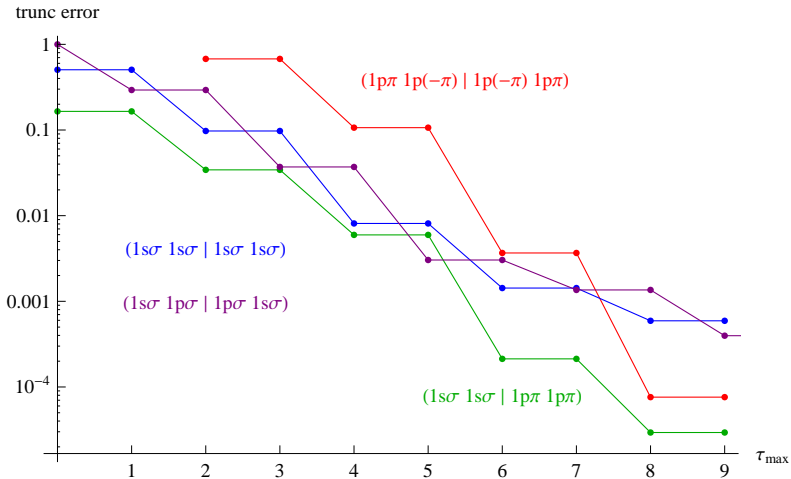}
\caption{Estimated relative truncation error of the sum over $\tau = \nu,\dots,\tau_{\max}$ in the Neumann expansion~\eqref{eq:NeumannExpansion}, exemplified by the Coulomb integrals~\eqref{eq:CoulombInt}. Wavefunctions are taken from section~\ref{sec:Framework} with $R = 121\,\mathrm{pm}$ (experimental bond length of $^{16}\mathrm{O}_2$). The observed exponential convergence renders the expansion~\eqref{eq:NeumannExpansion} particularly useful. We have calculated the error by comparison with the sum up to $\tau = 11$. Only every second $\tau$ contributes due to the symmetry constraint~\eqref{eq:AngularSymmetryConstraint} below, which explains the plateaus of the curves.}
\label{fig:TauConv}
\end{figure}

\subsection{Angular Coulomb Integral}

The expansion~\eqref{eq:NeumannExpansion} admits a separation of the $\eta_1$ and $\eta_2$ integrals. Both are of the same form, so it suffices to restrict the following presentation to the $\eta_1$ integral. Taking into account the volume element, we have to calculate
\begin{equation}
\label{eq:AngularIntegral}
\int_{-1}^1 S^{\mu}_{\ell}(\ii p,\Delta q,\eta) S^{\mu'}_{\ell'}(\ii p',\Delta q,\eta) P^\nu_\tau(\eta)\,\eta^j \ud\eta
\end{equation}
for $j \in \{0,2\}$, where we have once again set $\mu := \abs{m}$ and $\mu' := \abs{m'}$. Since $S^\mu_{\ell}(\ii p,0,\eta)$ and $P^\nu_\tau(\eta)$ have parity  $(-1)^{\ell+\mu}$ and $(-1)^{\tau+\nu}$, respectively, it follows that in the homonuclear case, \eqref{eq:AngularIntegral} is non-zero only if
\begin{equation}
\label{eq:AngularSymmetryConstraint}
\ell+\mu+\ell'+\mu'+\tau+\nu\quad \text{is even}.
\end{equation}
Plugging the expansion~\eqref{eq:SpheroidalAngularExpansion} into~\eqref{eq:AngularIntegral} results in a linear combination of integrals of the following form (see also~\cite[Appendix D]{AubertBessis1974}), which are explicitly solved by Wigner 3j symbols:
\begin{equation*}
\begin{split}
&\frac{1}{2} \prod_{i=1}^3 \left(\frac{(\ell_i-\mu_i)!}{(\ell_i+\mu_i)!}\right)^{1/2} \int_{-1}^1 P^{\mu_1}_{\ell_1}(\eta) P^{\mu_2}_{\ell_2}(\eta) P^{\mu_3}_{\ell_3}(\eta) \ud\eta\\
&\quad = \begin{pmatrix}\ell_1&\ell_2&\ell_3\\0&0&0\end{pmatrix} \times
\begin{cases}
(-1)^{\mu_2+\mu_3} \begin{pmatrix}\ell_1&\ell_2&\ell_3\\\mu_1&-\mu_2&-\mu_3\end{pmatrix}\vspace{5pt}& \text{if $\mu_3 = \abs{\mu_1-\mu_2}$}\\
(-1)^{\mu_3} \begin{pmatrix}\ell_1&\ell_2&\ell_3\\\mu_1&\mu_2&-\mu_3\end{pmatrix}& \text{if $\mu_3 = \mu_1+\mu_2$}
\end{cases}
\end{split}
\end{equation*}
The equation is valid for non-negative integers $\mu_1, \mu_2, \mu_3$, assuming (w.l.o.g.) $\mu_1 \ge \mu_2$. Given the Wigner 3j symbols, this is much easier to calculate than Gaunt's formula of the integral.

The factor $\eta^j$ in~\eqref{eq:AngularIntegral} for $j = 2$ can be incorporated by the identity~\eqref{eq:LegendreMultX} above.

\subsection{Radial Coulomb Integral}

For conciseness of notation, we subsume the ``quantum numbers'' $n \ell m$ from~\eqref{eq:CoulombInt} as $i$, and equivalently for $i'$, $\tilde{i}$ and $\tilde{i}'$. The radial contribution to the Coulomb integral is computationally much more challenging due to the dependence of whether $\xi_1 < \xi_2$ or $\xi_1 \ge \xi_2$. Thus, the integrals over $\xi_1$ and $\xi_2$ cannot be separated any more; instead, we obtain the nested integrals
\begin{equation}
\label{eq:IntRadialNested}
\int_1^\infty \Lambda_{\tilde{i}}(\xi_2) \Lambda_{\tilde{i}'}(\xi_2) Q^\nu_\tau(\xi_2) \xi_2^{\tilde{j}} \int_1^{\xi_2} \Lambda_i(\xi_1) \Lambda_{i'}(\xi_1) P^\nu_\tau(\xi_1)\,\xi_1^j \ud\xi_1 \ud\xi_2 + \left\langle i i'j \leftrightarrow \tilde{i}\tilde{i}'\tilde{j}\right\rangle
\end{equation}
for $j, \tilde{j} \in \{0,2\}$ due to the volume element.

The authors~\cite{AubertBessis1974} apply an integral transformation (from Ref.~\cite{Ruedenberg1951}) to~\eqref{eq:IntRadialNested} and then solve the resulting integral numerically. It consists of an outer integral over the product of two functions, which are themselves integrals. Although this approach inherently respects the symmetry $k \leftrightarrow \tilde{k}$, we haven't found it computationally advantageous as compared to solving~\eqref{eq:IntRadialNested} directly, since three integrals need to be calculated instead of two.

In the following, we provide the details of our approach. We employ the methods developed in section~\ref{sec:LaguerreExpansions} to evaluate~\eqref{eq:IntRadialNested}. As first (and most expensive) step, set $p_{ii'} := (p_i + p_{i'})/2$ and calculate $\vect{d}_{ii'}$ via~\eqref{eq:HylleraasMultArgscale} such that
\begin{equation}
\label{eq:LambdaMult}
\Lambda_i(\xi) \Lambda_{i'}(\xi) \equiv H^\mu_{\vect{d}_i}(2 p_i\,x) H^{\mu'}_{\vect{d}_{i'}}(2 p_{i'}\,x) \stackrel{!}{=} H^{\nu}_{\vect{d}_{ii'}}(2 p_{ii'}\,x), \quad x := \xi-1.
\end{equation}
Proceed analogously for $p_{\tilde{i}\tilde{i}'}$ and $\vect{d}_{\tilde{i}\tilde{i}'}$. Finally, set $p_{ii'\tilde{i}\tilde{i}'} := (p_{ii'} + p_{\tilde{i}\tilde{i}'})/2$ and calculate coefficients $\vect{b}_{ii'}$, $\vect{b}_{\tilde{i}\tilde{i}'}$ via~\eqref{eq:HylleraasArgscale} such that
\begin{equation*}
H^{\nu}_{\vect{d}_{ii'}}(2 p_{ii'}\,x) \stackrel{!}{=} H^{\nu}_{\vect{b}_{ii'}}(2 p_{ii'\tilde{i}\tilde{i}'}\,x)
\end{equation*}
(equivalently for $\vect{b}_{\tilde{i}\tilde{i}'}$). In case $p_{ii'} = p_{\tilde{i}\tilde{i}'}$, this step can be cut short by simply setting $\vect{b}_{ii'} := \vect{d}_{ii'}$ and $\vect{b}_{\tilde{i}\tilde{i}'} := \vect{d}_{\tilde{i}\tilde{i}'}$. Then, after a change of variables, the integral~\eqref{eq:IntRadialNested} (times the normalization factor $(\tau-\nu)!/(\tau+\nu)!$ and for $j,\tilde{j} = 0$) equals
\begin{equation}
\label{eq:RadialIntMatrix}
\hprod{\vect{b}_{ii'}}{B^\nu_\tau(z)\,\vect{b}_{\tilde{i}\tilde{i}'}}, \quad z := 2 p_{ii'\tilde{i}\tilde{i}'}
\end{equation}
with the matrix $B^\nu_\tau(z) \equiv \left(b^\nu_{\tau,k \tilde{k}}(z)\right)_{k\tilde{k}}$ defined by
\begin{multline}
\label{eq:RadialCoulombInt}
b^\nu_{\tau,k \tilde{k}}(z)
:= \frac{(\tau-\nu)!}{(\tau+\nu)!} \int_0^\infty H^\nu_{\tilde{k}}(\tilde{x}) Q^\nu_\tau(1+\tilde{x}/z)\\
\times \int_0^{\tilde{x}} H^\nu_k(x) P^\nu_\tau(1+x/z) \ud x \ud \tilde{x} + \left\langle k \leftrightarrow \tilde{k}\right\rangle.
\end{multline}
The factors $\xi_1^j$ and $\xi_2^{\tilde{j}}$ for $j = 2$ or $\tilde{j} = 2$ in the integral~\eqref{eq:IntRadialNested} can be incorporated via equation~\eqref{eq:HylleraasMultX}, similar to the angular integral.

Thus, given the matrix $B^\nu_\tau(z)$, we have reduced the rather expensive integral~\eqref{eq:IntRadialNested} to the simple matrix formula~\eqref{eq:RadialIntMatrix}. To obtain $B^\nu_\tau(z)$, we have first precomputed the entries~\eqref{eq:RadialCoulombInt} symbolically in $z$ as detailed below. Still, the resulting formulas are quite extensive and preclude a fast numerical evaluation. Our remedy consists in a Taylor expansion of (the entries in) $B^\nu_\tau(z)$,
\begin{equation}
\label{eq:BNuTauTaylor}
B^\nu_\tau(z) \approx \sum_{n=0}^{n_{\max}} \frac{\left(z-z_0\right)^n}{n!} B^{\nu\,(n)}_\tau(z_0).
\end{equation}
We precompute the derivatives $B^{\nu\,(n)}_\tau$ symbolically (up to $n_{\max} = 8$) and then evaluate them at (half)-integers $z_0 = 1, 1.5, 2, 2.5, \dots$. Due to potential numeric cancellation effects, we employ high-precision arithmetic for this intermediate step. Nevertheless, the entries of the resulting matrices $B^{\nu\,(n)}_\tau(z_0)$ are well-behaved and do not increase for higher values of $n$. These numeric matrices are then stored on disk for later usage.

\paragraph{Error estimation}

The sampling of half-integer evaluation points ensures that for each occurring $z$, there is a closest $z_0$ with $\abs{z-z_0} \le 1/4$. Thus, a very coarse error estimate of the Taylor expansion~\eqref{eq:BNuTauTaylor} gives an error of $10^{-11}$, when assuming that the individual entries of $B^{\nu\,(n)}_\tau$ are in the order of $1$, independent of $n$. In reality, we observe even better results, up to double floating-point precision $10^{-16}$.

Until now, we have not yet discussed the truncation error of the Laguerre expansions. As illustrated in figure~\ref{fig:DkExpansion} above, we can reach machine precision due to the exponential decay. However, the number of required coefficients depends on the particular decay parameters. When multiplying two expansions via~\eqref{eq:HylleraasMult}, these numbers typically add up to give the number of coefficients in the resulting expansion. Thus, in our setting, the coefficient vectors $\vect{b}_{ii'}$ and $\vect{b}_{\tilde{i}\tilde{i}'}$ from the formula~\eqref{eq:RadialIntMatrix} have approximately length $36$.

\paragraph{Cost analysis}

Summarizing the above steps after precomputation, our algorithm only needs the \emph{numeric} matrices $B^{\nu\,(n)}_\tau(z_0)$ from the Taylor expansion~\eqref{eq:BNuTauTaylor} as input, instead of the symbolic integrals~\eqref{eq:RadialCoulombInt}. In particular, no numeric integration is required.

The most expensive remaining step is the Laguerre expansion of the product $\Lambda_i(\xi) \Lambda_{i'}(\xi)$ in~\eqref{eq:LambdaMult}. Assuming that the expansion vectors $\vect{d}_i$ and $\vect{d}_{i'}$ have length $K$ and the resulting vector $\vect{d}_{ii'}$ length $2K$, the operation~\eqref{eq:HylleraasMultCoeff} has to be performed $2K$ times, leading to the asymptotic total cost $\mathcal{O}(K^3)$. In our setting, $K$ is typically equal to $18$. Since matrix operations are highly optimized, the computation time is in the order of milliseconds on modern PCs.

\paragraph{Symbolic calculation of the integrals~\eqref{eq:RadialCoulombInt}}

In what follows, we reduce~\eqref{eq:RadialCoulombInt} to the integrals~\eqref{eq:LaguerreNestedExpInt}, \eqref{eq:InvXZLaguerreInt} and~\eqref{eq:LaguerreNestedArcothInt} (solved in section~\ref{sec:LaguerreExpArcothIntegral}). We focus on the relevant cases $\nu = 0,1,2$, but our approach can easily be extended to higher $\nu$. In the simplest case $\tau, \nu = 0$, the integrals~\eqref{eq:LaguerreNestedArcothInt} and~\eqref{eq:RadialCoulombInt} coincide since $Q_0(\xi) = \mathrm{arcoth(\xi)}$. For general $\tau,\nu$, our strategy consists of ``absorbing'' the Legendre functions into the Laguerre polynomials from $H^\nu_k$ and $H^\nu_{\tilde{k}}$ by repeatedly applying equation~\eqref{eq:LaguerreMultX} (multiplication by $x$).

First, remember that the Legendre function of the second kind splits into
\begin{equation*}
\label{eq:QTauNuDecomp}
Q^\nu_\tau(\xi) = \frac{G^\nu_\tau(\xi)}{\left(\xi^2-1\right)^{\nu/2}} + P^\nu_\tau(\xi)\,\mathrm{arcoth}(\xi),
\end{equation*}
where $P^\nu_\tau(\xi)$ is the Legendre function of the first kind and $G^\nu_\tau$ a polynomial of order $\tau-1+\nu$.

Consider the case $\nu = 0$: Since $P_\tau(\xi)$ is actually a polynomial, repeated application of~\eqref{eq:LaguerreMultX} allows us to write
\begin{equation}
\label{eq:LegendreAbsorb}
L_k(x) P_\tau(1+x/z) = \sum_{i=k-\tau}^{k+\tau} a_{k \tau i}(z) L_i(x)
\end{equation}
with some coefficients $a_{k \tau i}(z)$. Proceeding similarly with $G^\nu_\tau(1+x/z)$, the integral~\eqref{eq:RadialCoulombInt} becomes as a linear combination of the integrals~\eqref{eq:LaguerreNestedExpInt} and~\eqref{eq:LaguerreNestedArcothInt}.

The case $\nu = 1$ is more involved since $P^\nu_\tau(\xi)$ is no polynomial any more, but consists of the factor $(\xi^2-1)^{1/2}$ $\times$ a polynomial of order $\tau-1$. To circumvent this difficulty, we first rewrite
\begin{equation*}
H^1_k(x) P^1_\tau(1+x/z) = \frac{1}{\sqrt{2z+x}} H^1_k(x) \cdot \sqrt{2z+x}\, P^1_\tau(1+x/z).
\end{equation*}
A symbolic Laguerre expansion of the first factor (computed in~\eqref{eq:SqrtHylleraasInt}, section~\ref{sec:LaguerreExpArcothIntegral}) transforms the right hand side to a linear combination of
\begin{equation*}
H^1_i(x) \sqrt{2z+x}\, P^1_\tau(1+x/z).
\end{equation*}
Plugging in the definition of $H^1_i(x)$, we obtain
\begin{equation*}
\sqrt{1/(i+1)}\, L^1_i(x)\, \mathrm{e}^{-x/2} \underbrace{\sqrt{x\,(2z+x)}\, P^1_\tau(1+x/z)}_{\mathrm{poly}(x)}.
\end{equation*}
Since the second part is a polynomial in $x$, we can proceed as for $\nu = 0$. The same transformation works for $Q^\nu_\tau(\xi)$ as well, which completes the case $\nu = 1$.

Finally, for $\nu = 2$, the decomposition~\eqref{eq:QTauNuDecomp} times the factor $x$ (from $H^2_k(x)$) reads
\begin{equation*}
x\, Q^2_\tau(1+x/z) = \frac{z^2}{2z+x} G^2_\tau(1+x/z) + x\,P^2_\tau(1+x/z)\,\mathrm{arcoth}(1+x/z).
\end{equation*}
After a transformation similar to~\eqref{eq:LegendreAbsorb}, we conclude that the integral~\eqref{eq:RadialCoulombInt} can be reduced to a linear combination of the integrals~\eqref{eq:LaguerreNestedArcothInt} and~\eqref{eq:InvXZLaguerreInt} (after a change of variables $x \to x/2$ and using the Laguerre product coefficients $c^\vect{\mu}_\vect{i}(1)$ in~\eqref{eq:cmi_z}).

\paragraph{Testing the implementation}

As first check, we calculate the following (homonuclear) Coulomb integral and obtain
\begin{equation*}
\bprod{\psi_{\mathrm{1s\sigma}} \psi_{\mathrm{1s\sigma}}}{\psi_{\mathrm{1s\sigma}} \psi_{\mathrm{1s\sigma}}} = 0.780883 \quad \text{for} \quad R = 1.4\,\mathrm{a.u.},\quad Z_a = Z_b = 1
\end{equation*}
with summation up to $\tau_{\max} = 9$ in~\eqref{eq:NeumannExpansion}. This agrees to all digits with the tabulated value in Ref.~\cite{AubertBessis2nd1974}. Similarly, for the heteronuclear case ($\mathrm{HeH}^+$ molecular ion with $Z_a = 2$ and $Z_b = 1$), we obtain
$\bprod{\psi_{\mathrm{1s\sigma}} \psi_{\mathrm{1s\sigma}}}{\psi_{\mathrm{1s\sigma}} \psi_{\mathrm{1s\sigma}}} = 1.23207$, which agrees in $4$ digits with the value $1.23225$ from Ref.~\cite{AubertBessis2nd1974}. The discrepancy could stem from lower precision arithmetic in~\cite{AubertBessis2nd1974}, or from a different truncation of the Neumann expansion. We have verified all digits of our value using Mathematica's numeric integration routines to solve the integrals \eqref{eq:AngularIntegral} and~\eqref{eq:IntRadialNested} directly (which is much slower in this case), and by comparing the truncations $\tau_{\max} = 6, 7, 8, 9$ of the Neumann expansion (all agreeing in the first $7$ digits).

For further testing, we have numerically computed the integrals~\eqref{eq:IntRadialNested} with several other parameters, and found that the values agree in at least $12$ digits.

\section{Application to Diatomic Molecules}
\label{sec:Results}

To demonstrate the feasibility of our algorithmic approach, we apply it to the diatomic molecules $\mathrm{O}_2$ and $\mathrm{CO}$.

The $N$-electron Hamiltonian for diatomic molecules in atomic units (Born-Oppenheimer approximation) reads
\begin{equation}
\label{eq:HamiltonianNBody}
H = \sum_{i=1}^N \left(-\frac{1}{2}\Delta_i - \frac{Z_a}{r_{ia}} - \frac{Z_b}{r_{ib}}\right) + \sum_{i<j} \frac{1}{r_{ij}} + \frac{Z_a Z_b}{R},
\end{equation}
where $r_{ia}$ and $r_{ib}$ denote the distances of the $i^\mathrm{th}$ electron to the fixed nuclei at $(0,0,\mp R/2)$, respectively, $Z_a, Z_b \in \N_{>0}$ the nuclear charges (as for the single electron Schr\"odinger equation~\eqref{eq:SchroedingerSingle}), and $r_{ij} \equiv \abs{\vect{x}_i - \vect{x}_j}$ the inter-electron distance between electron $i$ and $j$. The first sum (denoted $H_0$) contains precisely the single-electron Hamiltonian~\eqref{eq:SchroedingerSingle}, the second sum is the inter-electron Coulomb repulsion (denoted $V_\mathrm{ee}$), and the last term the repulsion of the nuclei. The homonuclear version ($Z_a = Z_b =: Z$) describes atomic dimers like hydrogen $\mathrm{H}_2$ or oxygen $\mathrm{O}_2$.

Analogous to Ref.~\cite{NuclearChargeLimit2009}, it is instructive to investigate the limit of large nuclear charge $Z$. (We consider the homonuclear case here for simplicity.) Namely, a short calculation shows that if $\Psi(\vect{x}_1,\sigma_1,\dots,\vect{x}_N,\sigma_N) \in L^2_a((\R^3 \times \{\pm\tfrac12\})^N)$ solves the $N$-electron Schr\"odinger equation $H\Psi = E\Psi$ with $H$ defined in~\eqref{eq:HamiltonianNBody}, then the rescaled wavefunction
\begin{equation*}
\tilde{\Psi}(\vect{y}_1,\sigma_1,\dots,\vect{y}_N,\sigma_N) := Z^{-3N/2}\, \Psi(Z^{-1}\vect{y}_1,\sigma_1,\dots,Z^{-1}\vect{y}_N,\sigma_N)
\end{equation*}
solves
\begin{equation}
\label{eq:HamiltonianNBodyZ}
\left(\tilde{H}_0 + \frac{1}{Z} V_{ee} + \frac{1}{R}\right)\tilde{\Psi} = \frac{E}{Z^2} \tilde{\Psi}
\end{equation}
with
\begin{equation*}
\tilde{H}_0 := \sum_{i=1}^N \left(-\frac{1}{2} \Delta_i - \frac{1}{\abs{\vect{y}_i + \tfrac{1}{2} Z R\,\vect{e}_3}} - \frac{1}{\abs{\vect{y}_i - \tfrac{1}{2} Z R\,\vect{e}_3}}\right).
\end{equation*}
As $Z \to \infty$, $\tilde{H}_0$ describes two isolated atoms and the electron-electron interaction $\tfrac{1}{Z} V_{ee}$ becomes small due to the prefactor $\tfrac{1}{Z}$. Since $\tilde{H}_0$ depends on $Z$, we cannot repeat the exact same analysis as in Ref.~\cite{NuclearChargeLimit2009}, but our single-electron wavefunctions are eigenfunctions of $H_0$ nevertheless. Thus we expect that our calculations match highly-charged (electrically confined) molecular ions well and could serve as benchmark for alternative computational approaches.

To allow for comparison with experimental data, we focus on the paramagnetic ``triplet'' oxygen molecule $\mathrm{O}_2$ in the following paragraph, i.e., $Z = 8$ and $N = 16$. The ground state symmetry is characterized by the molecular symbol $^3 \Sigma_g^-$. That is, the spin quantum number equals $1$ (hence ``triplet''), the angular $L_z$ momentum quantum number is zero (rotation about internuclear axis), and the parity is even.

The common textbook version of the electronic quantum state reads as follows. All molecular orbitals up to $\pi_u(2p_{x,y})$ are completely filled (see figure~\ref{fig:MolecularOrbitals}), leaving the two remaining electrons in the ``antibonding'' $\pi_g^*\,(2p_{x,y})$ orbitals. These two electrons form a spin triplet, hence the paramagnetism. With the mapping from figure~\ref{fig:MolecularOrbitals}, the electronic configuration corresponds to the following Slater determinant:
\begin{equation*}
\Psi_1 = \ket{\mathrm{
1s\sigma_g\!\uparrow\downarrow
1p\sigma_u\!\uparrow\downarrow
2s\sigma_g\!\uparrow\downarrow
2p\sigma_u\!\uparrow\downarrow
1d\sigma_g\!\uparrow\downarrow
1p(\pm\pi)_u\!\uparrow\downarrow
1d(\pm\pi)_g\!\uparrow}}.
\end{equation*}

In what follows, we try to approximate the groundstate energy of $\mathrm{O}_2$ via the methods from the previous chapters, with summation up to $\tau_{\max} = 9$ in the Neumann expansion~\eqref{eq:NeumannExpansion}. We include all wavefunctions of the $^3 \Sigma_g^-$ symmetry subspace, restricted to configurations with the $\mathrm{1s\sigma_g\!\uparrow\downarrow}$ and $\mathrm{1p\sigma_u\!\uparrow\downarrow}$ orbitals completely filled, and the occupations of the higher orbitals (up to $\mathrm{1f\sigma_u\!\uparrow\downarrow}$) allowed to vary. In our case, this gives $54$ wavefunctions, including $\Psi_1$. For example, another state in the symmetry subspace reads
\begin{equation*}
\Psi_2 = \ket{\mathrm{\,\cdots\,
1p(\pm\pi)_u\!\uparrow
1d(\pm\pi)_g\!\uparrow\downarrow}},
\end{equation*}
which is the same as $\Psi_1$ except for half-occupied $\mathrm{1p(\pm\pi)_u}$ molecular orbitals instead of $\mathrm{1d(\pm\pi)_g}$.

Thus, the groundstate energy is the smallest eigenvalue of the $54 \times 54$ matrix $\hprod{\Psi_i}{H\,\Psi_j}_{i,j}$, with the Hamiltonian $H$ defined in~\eqref{eq:HamiltonianNBody}. Since the $\Psi_i$ are exact eigenstates of the $N$-body Hamiltonian without the inter-electron Coulomb repulsion $V_\mathrm{ee}$, the latter can be regarded as perturbation of $(H - V_\mathrm{ee})$ (see also equation~\eqref{eq:HamiltonianNBodyZ}).

We use the software toolbox~\cite{FermiFabSoftware,FermiFabPaper2011} to express $\hprod{\Psi_i}{V_\mathrm{ee}\,\Psi_j}$ as linear combination of Coulomb integral symbols~\eqref{eq:CoulombInt}, after tracing-out the spin variables. The symmetry properties $\bprod{a b}{c d} = \bprod{c d}{a b}$ and $\bprod{a b}{c d} = \conj{\bprod{b a}{d c}}$ simplify the resulting expressions. As concrete example, the following off-diagonal matrix element reads
\begin{equation*}
\hprod{\Psi_1}{V_{\mathrm{ee}} \Psi_2} = \bprod{\psi_{\mathrm{1p\pi_u}} \psi_{\mathrm{1d\pi_g}}}{\psi_{\mathrm{1p(-\pi)_u}} \psi_{\mathrm{1d(-\pi)_g}}} - \bprod{\psi_{\mathrm{1p\pi_u}} \psi_{\mathrm{1d(-\pi)_g}}}{\psi_{\mathrm{1p(-\pi)_u}} \psi_{\mathrm{1d\pi_g}}}.
\end{equation*}
Both diagonal entries $\hprod{\Psi_i}{V_{\mathrm{ee}} \Psi_i}$ ($i = 1,2$) are quite extensive, consisting of $79$ individual Coulomb integrals.

Our approach can easily be adapted to other symmetry subspaces. Thus we include the experimentally next low-lying symmetry levels $^1 \Delta_g$ and $^1 \Sigma_g^+$ as well (see for example Ref.~\cite{Bernath2002} for an overview).

\begin{figure}[!ht]
\centering
\includegraphics[width=0.9\textwidth]{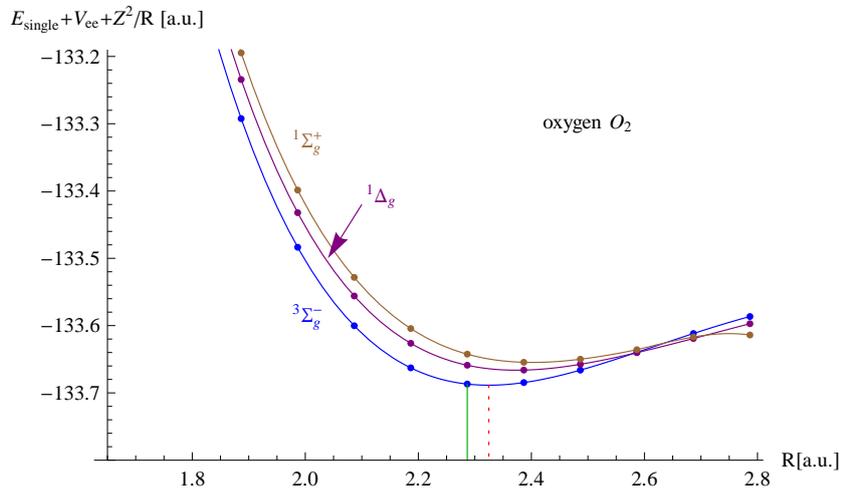}
\caption{Potential energy curves of the $\mathrm{O}_2$ molecule (this paper), i.e., lowest eigenvalue of the matrix $\hprod{\Psi_i}{H\,\Psi_j}$ with the Hamiltonian $H$ from~\eqref{eq:HamiltonianNBody} restricted to the corresponding symmetry subspace. Dots are calculated values, and the continuous line a spline interpolation of degree 3. The minimum $-133.689\,\mathrm{a.u.}$ of the $^3 \Sigma_g^-$ symmetry level is attained at $R = 2.325\,\mathrm{a.u.}$ (dotted red line). For comparison, the experimental bond length from the literature (green line) is also shown.}
\label{fig:O2Energy}
\end{figure}
The result of our calculations is plotted in figure~\ref{fig:O2Energy}, showing the electronic groundstate energy (blue dissociation curve) as well as excited energy levels (purple and brown curves) dependent of the nuclear distance $R$. Our computation predicts an optimal bond length $R_{\min} = 2.325\,\mathrm{a.u.}$ (dotted red line), which is quite close to the experimental value from the literature~\cite{OxygenMolecularConstants1975,HuberHerzberg1979}, $R_{\mathrm{exp}}(^{16}\mathrm{O}_2) = 121\,\mathrm{pm} = 2.2866\,\mathrm{a.u.}$ (green line). Additionally, we reproduce the experimental ordering of the symmetry states.

Having obtained the groundstate energy, we can calculate the dissociation energy $\mathrm{O}_2 \to 2\,\mathrm{O}$ by subtracting ($2 \times$) the energy of an individual oxygen atom. Since the outcome of theoretical calculations depends on the particular model (e.g., the Ansatz space of single-electron wavefunctions), similar models should be used for both the $\mathrm{O}_2$ molecule and the individual atoms. In our case, a close match regarding single atoms is Ref.~\cite{NuclearChargeLimit2009} as already mentioned above. Namely, the authors use hydrogen-like wavefunctions as Ansatz space and treat the inter-electron Coulomb repulsion as perturbation (similar to the present study). Additionally, the electronic configurations match ours in the $R \to \infty$ limit (available atomic subshells $\mathrm{1s}, \mathrm{2s}, \mathrm{2p}$, with the lowest $\mathrm{1s}$ subshell always occupied). From~\cite{NuclearChargeLimit2009}, $E_{\mathrm{O},\min} = -66.7048\,\mathrm{a.u.}$ for the groundstate angular momentum/spin symmetry $^3P$. Thus, we obtain the dissociation energy
\begin{equation}
\label{eq:Ediss_calc}
2\,E_{\mathrm{O},\min} - E_{{\mathrm{O}_2},\min} = 0.278971\,\mathrm{a.u.}
\end{equation}
For comparison, the experimental dissociation energy of oxygen (enthalpy change at $0\,\mathrm{K}$) is $E_{\mathrm{O}_2,\mathrm{exp\,diss}} = 5.1157\,\mathrm{eV} = 0.1879\,\mathrm{a.u.}$~\cite{BondDissociation1970}, which differs from our calculated value~\eqref{eq:Ediss_calc} by approximately $50\%$. The discrepancy is likely due to the small dimension of the Ansatz space (number of single-electron wavefunctions, up to the $2p$ subshell in our case). Note that the dissociation energy is $3$ orders of magnitude smaller than the total energy. Thus, subtracting groundstate energies requires at least $4$ correct decimal digits for just $1$ digit of the dissociation energy. In any case, our calculated value reproduces the experimental data qualitatively correct, in particular the sign (i.e., the fact that $\mathrm{O}_2$ binds).

To provide an example for a \emph{heteronuclear} molecule, we repeat the analogous calculations for carbon monoxide $\mathrm{CO}$, i.e., $Z_a = 8$, $Z_b = 6$ and $N = 14$. Figure~\ref{fig:COEnergyR} shows the resulting ground state dissociation curve with the same spheroidal Ansatz space (up to $1f\sigma\!\uparrow\downarrow$) as for oxygen. Notably, the deviation between the experimental bond length $R_\mathrm{exp}(^{12}\mathrm{C}^{16}\mathrm{O}) = 112.8\,\mathrm{pm}$ (green line, \cite{Bunker1970,HuberHerzberg1979}) and the calculated minimizer of the curve (dotted red line) is relatively large. This is presumably due to the small number of spheroidal basis functions. Indeed, when including the $3s\sigma\!\uparrow\downarrow$ spheroidal orbitals, the minimizer of the curve approaches the experimental value (figure~\ref{fig:COEnergyR3s}).

\begin{figure}[!ht]
\centering
\subfloat[]{
\includegraphics[width=0.5\textwidth]{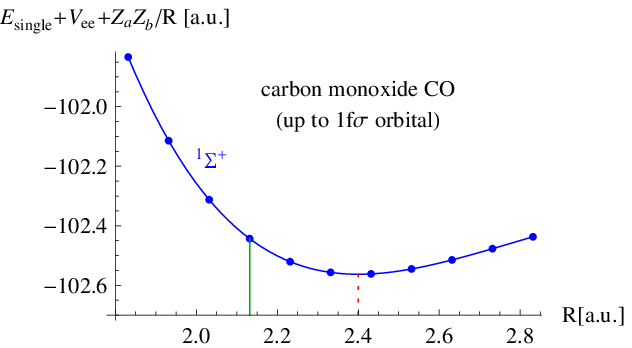}
\label{fig:COEnergyR}}
\subfloat[]{
\includegraphics[width=0.5\textwidth]{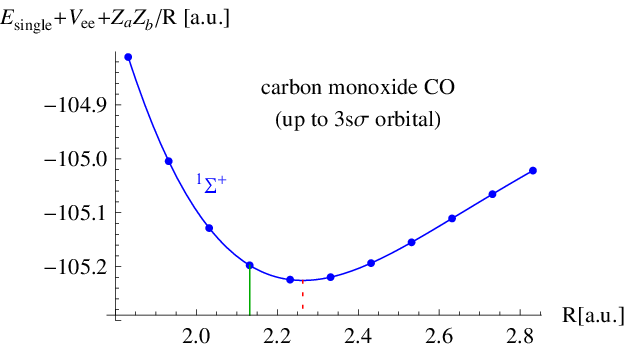}
\label{fig:COEnergyR3s}}
\caption{Potential energy curve of the $\mathrm{CO}$ molecule restricted to the groundstate $^1\Sigma^+$ symmetry subspace (this paper). (a) Same basis set as in figure~\ref{fig:O2Energy}; (b) additionally including the $3s\sigma\!\uparrow\downarrow$ spheroidal orbitals. In (b), the minimizer of the curve (dotted red line) is closer to the experimental bond length (green line). Note that the energy axes are shifted by approximately $3\,\mathrm{a.u.}$}
\end{figure}

\section{Conclusions and Outlook}
\label{sec:Conclusions}

We have developed and implemented an efficient computational framework to evaluate the angular and radial Coulomb/exchange integrals in prolate spheroidal coordinates by employing Neumann's expansion of $1/\abs{\vect{x}-\vect{y}}$ and taking advantage of symbolic integration as far as possible. The algorithm strongly relies on matrix operations to speed up computations.

A particular advantage of our approach is the universality of the precomputed numeric matrices in~\eqref{eq:BNuTauTaylor}. Once obtained, these matrices can be reused for subsequent calculations.

The application to the oxygen and carbon monoxide molecules shows the feasibility of our algorithm. We reproduce qualitatively correct energy curves, and the calculated bond length and dissociation energy are in reasonable agreement with experimental values.

A long-term goal of the present paper is a better understanding and quantitative description of atomic interactions and chemical bonds, which could be modeled using spheroidal orbitals. To reduce complexity, one could employ the well-known hydrogen-like orbitals for the core electrons (close to the nucleus). This combination of spheroidal and hydrogen-like orbitals requires proper orthonormalization and the calculation of Coulomb/exchange integrals between these different kind of orbitals. Inversing the LCAO Ansatz to approximate the spheroidal wavefunctions locally (close to an atomic nucleus) might be feasible for these purposes.

Finally, the algorithm presented here could be combined with established computational chemistry methods (like Configuration Interaction or Coupled Cluster) in future projects.

\paragraph*{Acknowledgements}

I'd like to thank Gero Friesecke, Ben Goddard and Martin F\"urst for many helpful discussions.

{\footnotesize

}

\end{document}